\documentclass[journal]{IEEEtran}
\usepackage[T1]{fontenc}
\usepackage{float}
\usepackage{footmisc}
\usepackage{enumerate}
\usepackage{latexsym}
\usepackage{amssymb}
\usepackage{bbm}
\usepackage{amsmath}
\usepackage{leftidx}
\usepackage{mathrsfs}
\usepackage{multirow}
\usepackage{slashbox}
\usepackage{threeparttable}
\usepackage{amsthm}
\usepackage{background}
\usepackage{algorithm}
\usepackage{algorithmic}
\usepackage{pifont}
\usepackage[noadjust]{cite}
\theoremstyle{definition}

\theoremstyle{plain}
\newtheorem{thm1}{Theorem}
\newtheorem{thm2}{Theorem}
\newtheorem{thm3}{Theorem}
\newtheorem{prop}[thm1]{Proposition}
\newtheorem{lemma}[thm2]{Lemma}
\newtheorem{cor}[thm3]{Corollary}
\begin{document}
\SetBgContents{}

\title{Optimal QoS-Aware Channel Assignment in D2D Communications with Partial CSI}
\author{\IEEEauthorblockN{Rui Wang,~\IEEEmembership{Student Member,~IEEE}, Jun Zhang,~\IEEEmembership{Senior Member,~IEEE}, S.H. Song,~\IEEEmembership{Member,~IEEE}, and Khaled B. Letaief,~\IEEEmembership{Fellow,~IEEE}}
\thanks{Manuscript received July 21, 2015; revised February 8, 2016 and June
27, 2016; accepted August 17, 2016. This work was supported by the Hong
Kong Research Grants Council Grant No. 610113. Part of this work has been presented at IEEE Globecom, San Diego, CA, Dec. 2015. The associate editor coordinating the
review of this paper and approving it for publication was A. Abrardo.}
\thanks{R. Wang, J. Zhang and S.H. Song are with the Department of Electronic and Computer Engineering, Hong Kong University of Science and Technology, Clear Water
Bay, Kowloon, Hong Kong (e-mail: rwangae@ust.hk; eejzhang@ust.hk;  eeshsong@ust.hk).}
\thanks{K. B. Letaief is with the Department of Electronic and Computer
Engineering, Hong Kong University of Science and Technology, Clear Water
Bay, Kowloon, Hong Kong, and also with Hamad bin Khalifa University, Doha,
Qatar (e-mail: eekhaled@ust.hk; kletaief@hkbu.edu.qa).}}

\maketitle
\begin{abstract}
  In this paper, we propose effective channel assignment algorithms for network utility maximization in a cellular network with underlaying device-to-device (D2D) communications. A major innovation is the consideration of partial channel state information (CSI), i.e., the base station (BS) is assumed to be able to acquire `partial' instantaneous CSI of the cellular and D2D links, as well as, the interference links. In contrast to existing works, multiple D2D links are allowed to share the same channel, and the quality of service (QoS) requirements for both the cellular and D2D links are enforced. We first develop an optimal channel assignment algorithm based on dynamic programming (DP), which enjoys a much lower complexity compared to exhaustive search and will serve as a performance benchmark. To further reduce complexity, we propose a cluster-based sub-optimal channel assignment algorithm. New closed-form expressions for the expected weighted sum-rate and the successful transmission probabilities are also derived. Simulation results verify the effectiveness of the proposed algorithms. Moreover, by comparing different partial CSI scenarios, we observe that the CSI of the D2D communication links and the interference links from the D2D transmitters to the BS significantly affects the network performance, while the CSI of the interference links from the BS to the D2D receivers only has a negligible impact.

\end{abstract}
\IEEEpeerreviewmaketitle

\section{Introduction}

\subsection{Motivation and Related Works}

The future fifth generation (5G) cellular networks are expected to achieve a $1000 \times$ higher area capacity, a roughly 1 ms roundtrip latency, and a $100 \times$ reduced cost per bit \cite{5G}. Device-to-device (D2D) communications is a promising approach to improve spectral efficiency and network coverage, as well as lower delays and power consumption, and it has been considered as one of the key technologies in 5G networks \cite{design,d2d-5g,conf}. Unlike the first four generations of cellular networks, where all communication links are forced to be routed via base stations (BSs), mobile users may communicate directly with each other through D2D links, under the control of the BSs in cellular networks \cite{d2d-5g,D2D,blue}. However, without effective network operation, the underlaying D2D communications may generate severe interference to the existing cellular network, and suffer strong interference from cellular users. Thus, efficient resource allocation schemes for controlling interference will be crucial for D2D communications \cite{survey,D2DSC}. Different from conventional cellular networks \cite{OFDM}, resource allocation in D2D communications faces new challenges with the presence of intra-cell interference \cite{design}, which bears similarities to cognitive radio (CR) systems \cite{cognitive1,cognitive2,cognitive3}. One major difference between the underlaying D2D communications and CR systems is that the D2D links are controlled by BSs, while the secondary users (SUs) in CR systems are acting spontaneously \cite{dcr1,dcr2}.

Resource allocation in D2D networks faces two main challenges. Firstly, when multiple D2D links are allowed to access the same channel, the dynamic resource allocation problem becomes an NP-hard problem \cite{nphard}. Secondly, in cellular networks with underlaying D2D communications, the total number of links, especially interference links, is usually very large. As a result, overwhelming overheads will be incurred for collecting the channel state information (CSI) of all these links \cite{partial-m}. Therefore, the commonly considered full CSI scenario is not practical, and the partial CSI case, where instantaneous CSI of part of the communication and interference links is unknown at the BS, should be considered. In this paper, we will develop effective channel assignment algorithms for D2D communications with partial CSI.

In order to simplify the NP-hard resource allocation problem, some existing works have assumed that at most one D2D link can access one channel \cite{onelink,2link,oneD2D,one2one}. In \cite{onelink} and \cite{2link}, the authors obtained the optimal resource sharing strategy to maximize the throughput with one cellular link and one D2D link, and an alternative greedy heuristic algorithm was proposed to maximize the system sum-rate in \cite{oneD2D}. Feng \emph{et al.} \cite{one2one} proposed an optimal resource sharing scheme based on maximum weighted bipartite matching. However, it cannot yield good spectral efficiency when at most one D2D link can access one channel, especially in a dense D2D network. Recent works have started to propose sub-optimal algorithms in the scenario where multiple D2D links may use the same channel. Zhang \emph{et al.} \cite{nphard,ingra} studied the recourse allocation problem in order to maximize the system throughput, adopting an interference-aware graph-based algorithm to get a sub-optimal solution. Then, a reverse iterative combinatorial auctions game was utilized to develop a resource sharing scheme to maximize the system sum-rate in \cite{ica1} and \cite{ICA}. However, these works did not take the reliability requirement of the D2D links into consideration. As the number of links sharing the same channel increases, cellular and D2D users may suffer poor performance even if they are admitted to the system. Thus, D2D networks where multiple D2D links are allowed to access the same channel should be investigated, and quality of service (QoS) requirements should be guaranteed for both cellular and D2D links.

Most previous works have focused on the full CSI scenario, but, recently, some papers have considered the partial CSI case. Considering that the BS cannot acquire CSI of the interference links between user devices, a maximum weighted bipartite matching algorithm was applied in \cite{feng2014qos} to get the optimal recourse allocation scheme in a D2D network. However, this work allowed at most one D2D link to access one channel. In \cite{c-d-p}, multiple D2D links were allowed to access the same channel with the assumption that the BS only has knowledge of the CSI of cellular links, and a centralized channel assignment algorithm and a distributed power control algorithm were developed in the high SINR region. However, the high SINR assumption does not usually hold in D2D communications. Furthermore, all the previous works considering partial CSI have only dealt with one particular partial CSI scenario, e.g., in \cite{feng2014qos}, the BS was assumed to know the CSI of all the links except the interference links between user devices. Thus, the relative importance of the CSI of different links in D2D networks is still unknown, and it is the question that will be explored in this paper.

\subsection{Contributions}

In this paper, we will investigate the channel assignment problem for maximizing the network utility in a cellular network with underlaying D2D communications. Different partial CSI scenarios will be considered to reveal the relative importance of the CSI for different links. To improve spectral efficiency, multiple D2D links may access the same channel. Moreover, to guarantee the per link performance, a minimum successful transmission probability requirement will be enforced for each active link. The major contributions of this paper are summarized as follows:
\begin{enumerate}
  \item To develop effective algorithms, explicit expressions are needed for the expected weighted sum-rate and successful transmission probabilities, which form the objective function and part of the constraints, respectively. We thus derive tractable expressions with different channel fading models, i.e., Rayleigh fading and Nakagami fading. In particular, with Rayleigh fading channels, closed-form expressions are obtained. For the full CSI scenario, these expressions can be simplified.
    \item To find the optimal solution for the NP-hard channel assignment problem, a dynamic programming (DP) algorithm is firstly proposed. This DP algorithm significantly reduces complexity compared to exhaustive search. While the complexity is still prohibitively high for a large system, it can well serve as a performance benchmark for systems with small to medium sizes. To the best of our knowledge, we are the first to propose an optimal algorithm for the channel assignment problem in D2D networks while multiple D2D links may share the same channel.
    \item To further reduce complexity, we propose a practical sub-optimal algorithm. The main difficulties come from the integer variables, and the non-convexity. Bipartite matching is adopted as the main tool to deal with the discrete variables. To define the weights in the bipartite graph according to the non-convex objective and constraints, we first divide all the links into several non-overlapping clusters, where links in the same cluster will share the same channel. A maximum weighted bipartite matching algorithm is then used to obtain the final channel assignment.
  \item Extensive simulation results are provided to demonstrate the effectiveness of the proposed algorithms and provide valuable design insights for D2D communications. In particular, the proposed sub-optimal algorithm outperforms existing ones, while approaching the optimal DP algorithm. It is demonstrated that the D2D links prefer sharing the uplink spectrum rather than the downlink spectrum. Moreover, by comparing four different partial CSI scenarios and the full CSI scenario, we observe that the knowledge of the CSI of the D2D communication links and the interference links from the D2D transmitters to the BS will have a significant effect on the performance, while the knowledge of the CSI of the interference links from the BS to the D2D receivers only has a negligible influence.
\end{enumerate}

\subsection{Organization}
The remainder of this paper is organized as follows. System model is presented in Section II. In Section III, a network utility maximization problem is formulated. The optimal DP algorithm is proposed in Section IV, while the cluster-based sub-optimal algorithm is proposed in Section V. The simulation results are shown in Section VI. Finally, Section VII concludes the paper.

\section{System Model}
\begin{figure}[!t]
  \centering
  \includegraphics[width=3.5in]{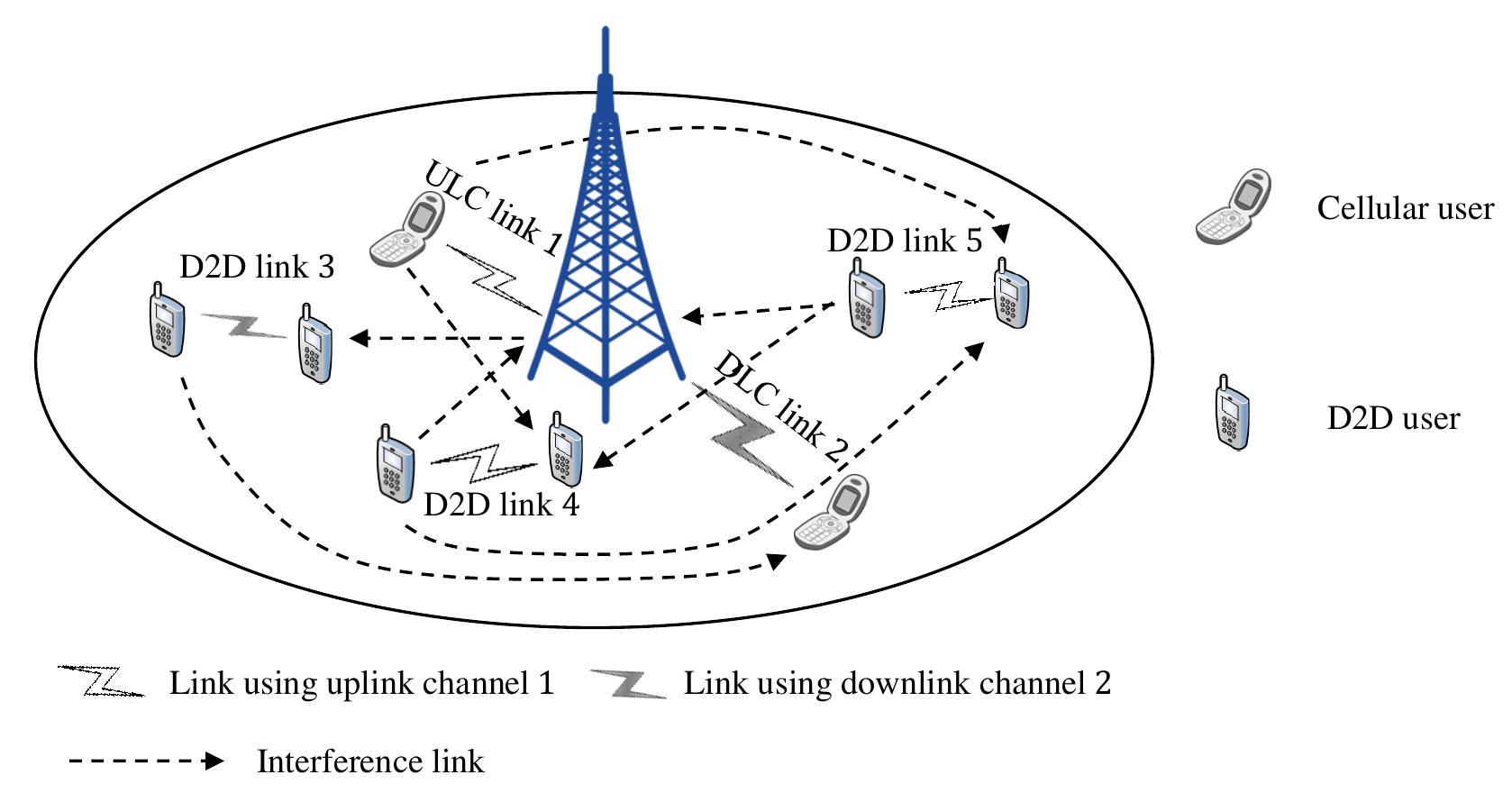}
  \caption{A sample network with two cellular links and three D2D links.}
  \label{model}
\end{figure}
In this section, we will first introduce assumptions on channel sharing among cellular and D2D links, and then the channel model will be presented.
\subsection{Channel Sharing Among Cellular and D2D Links}
As shown in Fig. \ref{model}, we consider a cellular network where cellular and D2D links coexist. Both uplink and downlink transmissions are considered. There are $N_{uc}$ uplink cellular (ULC) links, $N_{dc}$ downlink cellular (DLC) links, and $N_d$ D2D links, whose index sets are denoted as $\mathcal{C}_u = \{1,2,...,N_{uc}\}$, $\mathcal{C}_d = \{N_{uc}+1,N_{uc}+2,...,N_{uc}+N_{dc}\}$ and $\mathcal{D} = \{N_{uc}+N_{dc}+1,N_{uc}+N_{dc}+2,...,N_{uc}+N_{dc}+N_d\}$, respectively. Correspondingly, there are in total $N_c = N_{uc}+N_{dc}$ cellular links and $N = N_{c}+N_{d}$ communication links, whose index sets are denoted as $\mathcal{C} = \mathcal{C}_u \cup \mathcal{C}_d  = \{1,2,...,N_{c}\}$ and $\mathcal{S} = \mathcal{C} \cup \mathcal{D}  = \{1,2,...,N\}$, respectively. The transmit power of link $j \in \mathcal{S}$ is denoted as $p^t_j$. Orthogonal Frequency Division Multiple Access (OFDMA) is considered for both the cellular and D2D communications, as in an LTE network \cite{LTE}, and there are $M_u$ uplink channels and $M_d$ downlink channels, whose index sets are denoted as $\mathcal{CH}_u = \{1,2,...,M_u\}$ and $\mathcal{CH}_d = \{M_u+1,M_u+2,...,M_u+M_d\}$, respectively. Correspondingly, there are in total $M = M_{u}+M_{d}$ channels, whose index set is denoted as $\mathcal{CH} = \mathcal{CH}_u \cup \mathcal{CH}_d  = \{1,2,...,M\}$. Note that the channel resources are limited, it may happen that the cellular users cannot be served simultaneously. In this paper, we assume that cellular user admission control schemes have been applied, e.g., cellular users may be admitted in a round-robin manner to ensure fairness. Thus, the numbers of \emph{active} uplink and downlink cellular links will be no larger than the numbers of uplink and downlink channels, respectively.

Each channel may be shared by both cellular and D2D links, and multiple D2D links are allowed to share the same channel. Considering fairness and simplicity, it is assumed that each communication link, including both the cellular and D2D links, can access at most one channel. To guarantee the priority of cellular links, we assume that all the cellular links can access one channel, while some D2D links may not be allowed to access any channel, and under such circumstances, they are called \emph{inactive D2D links}. Correspondingly, D2D links which access one channel are called \emph{active D2D links}.

\subsection{Channel Model}
Assume that links $j \in \mathcal{S}$ and $z \in \mathcal{S}$ access the same channel $i \in \mathcal{CH}$. We shall use $g_{i,j}$ and $h_{i,z,j}$ to denote the channel gains of link $j$ and the interference link from link $z$ to link $j$, respectively. The channel gains contain the normalized small-scale fading, log-normal shadowing and distance based path loss. The interference channel gain from link $z$ to link $j$ using channel $i$ can be expressed as $h_{i,z,j}=C^{(p)} \beta_{i,z,j} \zeta_{z,j} \left( L_{z,j} \right) ^{-\alpha}$, where $C^{(p)}$ is the path loss constant, $L_{z,j}$ denotes the distance between the transmitter of link $z$ and the receiver of link $j$, $\alpha$ represents the path loss exponent, and $\beta_{i,z,j}$ and $\zeta_{z,j}$ imply small-scale fading and shadow fading, respectively. Similarly, the channel gain of link $j$ can be expressed as $g_{i,j}=C^{(p)} \beta_{i,j,j} \zeta_{j,j} \left( L_{j,j} \right) ^{-\alpha}$. The probability density function (pdf) of the small-scale variation of communication link $j$ is denoted as $f_{j}( \cdot )$, and the interference channel from link $z$ to link $j$ is assumed to be Nakagami fading with parameter $m_{z,j}$. Meanwhile, we assume that small-scale variations are independent among all the links and all the channels. Knowledge of the path loss and shadowing of all links is assumed to be available at the BS, while the instantaneous CSI of small-scale fading is only partially known at the BS, which will be specified later. Additive white Gaussian noise with zero mean and variance $\sigma^2$ is assumed at each receiver. Each of the cellular links and active D2D links is assumed to have a minimum QoS requirement, which is measured by the successful transmission probability, i.e., the probability that the received SINR is above a certain threshold.

\section{Problem Formulation}

In this paper, we investigate the channel assignment problem to maximize a certain network utility function with QoS guarantees for both cellular and D2D links. Let $\rho_{i,j}$, with $i \in \mathcal{CH}$ and $j \in \mathcal{S}$, denote the channel assignment for all cellular and D2D links, which is defined as
\begin{equation} \label{rho}
\rho_{i,j}=\left\{
                \begin{array}{ll}
                1 & \text{if channel } i \text{ is assigned to link } j,\\
                0& \text{otherwise}.
                \end{array}
                \right.
\end{equation}
The set $\mathcal{L}_i$ includes all the links accessing channel $i$, given by
$\mathcal{L}_i=\left\{j|j \in \mathcal{S} \text{ and } \rho_{i,j}=1\right\}$.
The received SINR of link $j$ using channel $i$ can then be given as
\begin{equation} \label{SINR}
\xi_{i,j} \left( \mathcal{L}_i \right) ={\frac{{p^t_j g_{i,j} }}
{{\sigma^2  + \sum\limits_{z \in \mathcal{L}_i \backslash \{j\} } {p^t_z h_{i,z,j} } }}}.
\end{equation}
For channel $i$, we adopt the corresponding utility function as a function of the set of links using channel $i$, denoted as $U_i \left( \mathcal{L}_i \right)$. Note that the function $U_i \left( \cdot \right)$ is an arbitrary function from the sets of links to real numbers. In the partial CSI case, the BS may not have enough CSI to calculate the instantaneous SINR for all links. Thus, each link is assumed to have a minimum successful transmission probability requirement, denoted as $\psi^{\text{min}}_{j}$, and the successful transmission probability is
$\text{Pr} \left[ \xi_{i,j}\left( \mathcal{L}_i \right) \geqslant \xi^{\text{min}}_{j} \right]$, denoted by $p^s_{i,j}$,
where $\xi^{\text{min}}_{j}$ denotes the minimum SINR requirement of link $j$ and $\text{Pr} \left[ \cdot \right]$ denotes probability.
Then, the utility maximization problem can be formulated as\footnote{Note that with this formulation there may be some users that cannot be served, i.e., they cannot meet the QoS requirements. This may cause a fairness issue, which can be relieved by assigning different weights for different users to control their priorities, which is adopted in this paper.}
\begin{align}
\mathop {\max }\limits_{\rho_{i,j} } &\sum\limits_{i \in \mathcal{CH}} { U_i \left( \mathcal{L}_i \right)} ,\label{problem} \\
\text{s.t. }&p^s_{i,j}=\text{Pr} \left[ \xi_{i,j}\left( \mathcal{L}_i \right) \geqslant \xi^{\text{min}}_{j} \right]  \geqslant \psi^{\text{min}}_{j} , \notag \\
&\forall j \in \mathcal{S} \text{ and } \sum\limits_{i \in \mathcal{CH}} {\rho _{i,j} } = 1, \tag{\ref{problem}a}\\
&\sum\limits_{j \in \mathcal{C}} {\rho _{i,j} } \leqslant 1, \rho _{i,j} \in \{0,1\}, \forall i \in \mathcal{CH}, \tag{\ref{problem}b} \\
&\sum\limits_{i \in \mathcal{CH}} {\rho_{i,j} } = 1, \rho _{i,j} \in \{0,1\}, \forall j \in \mathcal{C}, \tag{\ref{problem}c} \\ 
&\sum\limits_{i \in \mathcal{CH}} {\rho_{i,j} } \leqslant 1, \rho _{i,j} \in \{0,1\}, \forall j \in \mathcal{D}, \tag{\ref{problem}d} 
\end{align}
\begin{align}
&\sum\limits_{i \in \mathcal{CH}_d} {\rho_{i,j} } = 0, \rho _{i,j} \in \{0,1\}, \forall j \in \mathcal{C}_u, \tag{\ref{problem}e} \\
&\sum\limits_{i \in \mathcal{CH}_u} {\rho_{i,j} } = 0, \rho _{i,j} \in \{0,1\}, \forall j \in \mathcal{C}_d. \tag{\ref{problem}f}
\end{align}
Constraint (\ref{problem}a) guarantees the QoS requirements of both cellular and active D2D links. Constraint (\ref{problem}b) ensures that different cellular links cannot access the same channel. Constraints (\ref{problem}c) and (\ref{problem}d) imply that each cellular link can access one channel and each D2D link can access at most one channel, respectively, and constraints (\ref{problem}e) and (\ref{problem}f) guarantee that uplink cellular links cannot access downlink channels and downlink cellular links cannot access uplink channels. In the remainder of this section, with the expected weighted sum-rate as an example of the utility function, we will derive closed-form expressions of the utility function and QoS constraints.

\subsection{Expected Weighted Sum-Rate Maximization with Partial CSI}

With partial CSI, the expected weighted sum-rate will be used as the performance metric, where the weights can be adjusted to control the priorities and fairness of different users, e.g., cellular users may have a higher priority than D2D users. Same as \cite{feng2014qos}, we regard the user rate as zero when the outage occurs. The utility function of channel $i$ is then given as
\begin{equation} \label{weighted sr}
 U_i \left( \mathcal{L}_i \right)=\sum\limits_{j \in \mathcal{L}_i} w_j \mathbb{E} \left[ \log \left( 1+\xi_{i,j} \left( \mathcal{L}_i \right) \right) \left| \xi_{i,j}\left( \mathcal{L}_i \right) \geqslant \xi^{\text{min}}_{j} \right. \right]p^s_{i,j},
\end{equation}
where $w_j$ is the weight of link $j$ and $\mathbb{E} \left[ \cdot | \cdot \right]$ denotes the conditional expectation. We denote the expected rate of link $j$ using channel $i$ as $r_{i,j}=\mathbb{E} \left[ \log \left( 1+\xi_{i,j} \left( \mathcal{L}_i \right) \right) \left| \xi_{i,j}\left( \mathcal{L}_i \right) \geqslant \xi^{\text{min}}_{j} \right. \right]p^s_{i,j}$. In a cellular network with underlaying D2D communications, there exists five kinds of links, namely, the cellular communication links, the D2D communication links, the interference links between user devices, the interference links from the BS to the D2D receivers, and the interference links from the D2D transmitters to the BS. In order to compare the importance of the CSI of the different kinds of links, we consider four different scenarios with partial CSI. Table \ref{tab_4case} summarizes the four scenarios.

By allowing multiple D2D links to share the same channel, interference becomes complicated. Since the BS may only acquire the CSI for part of the interference links, we rewrite the SINR of link $j$ using channel $i$ as
\begin{equation} \label{SINRp}
\xi_{i,j} \left( \mathcal{L}_i \right) =
{\frac{{\lambda_{j,j} \beta_{i,j,j} }}
{{\sigma^2  + \sum\limits_{z \in \mathcal{L}_i \backslash \{j\} } {\lambda_{z,j} \beta_{i,z,j} } }}}=
\frac{\lambda_{j,j}\beta_{i,j,j} }
{{\nu + Y_{i,j} }},
\end{equation}
where $\beta_{i,z,j}$, with $z,j \in \mathcal{S}$, denotes the small scale fading gain, $\lambda_{j,j} \triangleq p^t_j g_{i,j}/ \beta_{i,j,j}$ and $\lambda_{z,j} \triangleq p^t_z h_{i,z,j}/ \beta_{i,z,j}$ are products of the transmit power and large scale fading gains. Let $\mathcal{L}'_i$ include all the links for which the BS cannot acquire their small-scale fading gains to link $j$. Then, the noise plus interference power can be divided into two parts, i.e., the sum-power of the noise and the interference links whose CSI is fully known at the BS, denoted by $\nu = \sigma^2 + \sum \limits_{z \in \mathcal{L}_i \backslash \{j\}-\mathcal{L}'_i} \lambda_{z,j} \beta_{i,z,j}$, and the power of interference links whose small scale fading gains are unknown at the BS, denoted by $Y_{i,j} ( \mathcal{L}'_i )=\sum\limits_{z \in \mathcal{L}'_i} {\lambda_{z,j}\beta_{i,z,j}} $. We derive the successful transmission probabilities and expected rates in \emph{Propositions} \ref{probability}, assuming arbitrary fading distributions for the signal links.
\begin{prop} \label{probability}
The successful transmission probability of link $j$ using channel $i$ is given by\footnote{For practical purpose, the infinite series in these expressions can be calculated accurately with finite terms, as shown in \cite{sumofgamma}. \label{footrepeat}}
\begin{align} \label{propnakagami}
p^s_{i,j} = 
\left\{
\begin{array}{ll}
\prod \limits_{z \in \mathcal{L}'_i} \left( \frac{\lambda_{z,j}}{m_{z,j}} \right)^{m_{z,j}}  , &  \multirow{2}{*}{$\text{if } |\mathcal{L}'_i|>0,$} \\
\times \sum \limits_{n=0}^{+ \infty} \frac{\delta_n (\theta^{\max}_j)^{-\rho} C_{i,j} }{\Gamma \left(\rho+n \right)} & \\
\multirow{2}{*}{$\mathbbm{1}[\lambda_{j,j} \beta_{i,j,j} / \nu \geqslant \xi_j^{\min}],$} & \text{if }  |\mathcal{L}'_i|=0 \\
& \text{and } \beta_{i,j,j} \text{ is known}, \\
\multirow{2}{*}{$\int_{\xi^{\min}_j \nu/ \lambda_{j,j}}^{\infty} f_j(x) dx,$}  & \text{if } |\mathcal{L}'_i|=0 \\
& \text{and } \beta_{i,j,j} \text{ is unknown}.
\end{array}
\right.
\end{align}
where $\theta^{\max}_j=\max \limits_{z \in \mathcal{L}'_i} \frac{ \lambda_{z,j}}{m_{z,j}}$, $\rho=\sum \limits_{z \in \mathcal{L}'_i} m_{z,j}$, $\mathbbm{1}[\lambda_{j,j} \beta_{i,j,j} / \nu \geqslant \xi_j^{\min}]$ equals $1$ if $\lambda_{j,j} \beta_{i,j,j} / \nu \geqslant \xi_j^{\min}$ and equals $0$ otherwise, and $\delta_n$ is given by (\ref{refdelta}) in Appendix A.
$C_{i,j}$ is given by
\begin{align} 
&C_{i,j}= \notag \\
&\left\{
\begin{array}{ll}
\gamma \left( \rho +n, \eta_{i,j} \theta^{\max}_j \right) & \text{if } \beta_{i,j,j} \text{ is known}, \\
\int_{ \frac{\xi^{\min}_j \nu}{\lambda_{j,j}} }^{\infty} \gamma \left( \rho +n, \left( \frac{\lambda_{j,j}x}{\xi^{\text{min}}_{j}}- \nu \right) \theta^{\max}_j \right) & \multirow{2}{*}{$\text{if } \beta_{i,j,j} \text{ is unknown,}$} \\
\times f_j(x) dx &
\end{array}
\right.
\end{align}
where $\eta_{i,j}=\max \left(0, \lambda_{j,j}\beta_{i,j,j}/\xi^{\text{min}}_{j}- \nu \right)$. The expected rate of link $j$ using channel $i$ is given by\footref{footrepeat}
\begin{align} \label{expectedrate}
&r_{i,j} = \notag \\
&\left\{
\begin{array} {ll}
 \log(e)
\prod \limits_{z \in \mathcal{L}'_i} \left( \frac{\lambda_{z,j}}{m_{z,j}} \right)^{m_{z,j}} & \multirow{2}{*}{$\text{if } |\mathcal{L}'_i|>0$,} \\
\times \sum \limits_{n=0}^{+ \infty} \frac{\delta_n \mu_{\rho+n-1} (\theta^{\max}_j)^{n}}{\Gamma \left(\rho+n \right)}, & \\
\log(1+\lambda_{j,j} \beta_{i,j,j}/\nu) & \text{if } |\mathcal{L}'_i|=0 \text{ and }  \\
\times \mathbbm{1}[\lambda_{j,j} \beta_{i,j,j} / \nu \geqslant \xi_j^{\min}], & \beta_{i,j,j} \text{ is known}, \\
\multirow{2}{*}{$\int_{\xi^{\min}_j \nu / \lambda_{j,j}}^{\infty} \log(1+\lambda_{j,j} x/\nu) f_j(x) dx,$}  & \text{if } |\mathcal{L}'_i|=0 \text{ and } \\
& \beta_{i,j,j} \text{ is unknown}.
\end{array}
\right.
\end{align}
where
\begin{align} \label{refmu}
&\mu_{\rho+n-1}= \notag \\
&\left\{
\begin{array}{ll}
\int_0^{\eta_{i,j}} \ln\left( 1+ \frac{\lambda_{j,j}\beta_{i,j,j}}{\nu+y}\right) & \multirow{2}{*}{$\text{if } \beta_{i,j,j} \text{ is known,}$} \\
y^{\rho+n-1} e^{ -y \theta_j^{\max} } dy, & \\
\int_{\frac{\xi^{\min}_j \nu}{\lambda_{j,j}}}^{\infty}\int_0^{\frac{\lambda_{j,j} x}{\xi^{\min}_{j}}-\nu} \ln\left( 1+ \frac{\lambda_{j,j}x}{\nu+y}\right) & \multirow{2}{*}{$\text{if } \beta_{i,j,j} \text{ is unknown.}$} \\
\times  y^{\rho+n-1} e^{ -y \theta_j^{\max} } f_j(x) dy dx,  &
\end{array}
\right.
\end{align}

\end{prop}
\begin{proof}
See Appendix A.
\end{proof}
Notice that a main difficulty to evaluate expressions of this result is the complicated integration in (\ref{refmu}). For the special case where all the fading parameters $m_{z,j}$ are integer numbers, $\mu_{\rho+n-1}$ in (\ref{refmu}) can be calculated more efficiently according to \emph{Lemma} \ref{expectmu}.
\begin{lemma} \label{expectmu}
Assuming the fading parameters $m_{z,j}$ are integer numbers, if the small scale fading gain of link $j$ is known at the BS, $\mu_k$ can be calculated as
\begin{align} \label{closedmu}
\mu_{0}=&D \left( \eta_{i,j},\nu+\lambda_{j,j} \beta_{i,j,j} \right)-D \left(\eta_{i,j},\nu \right), \notag \\
\mu_{k}=&\left(  \theta^{\max}_j \right)^{-1} \big[ E_k \left( \eta_{i,j}, \nu + \lambda_{j,j} \beta_{i,j,j} \right)- A_k \left( \eta_{i,j}, \lambda_{j,j} \beta_{i,j,j} \right) \notag \\
&- E_k \left( \eta_{i,j}, \nu \right) + k \mu_{k-1} \big], k=1,2,\cdots,
\end{align}
and if the small scale fading gain of link $j$ is unknown at the BS, $\mu_k$ can be calculated as
\begin{align}
\mu_{0}=&\int_{\xi^{\min}_j \nu / \lambda_{j,j}}^{\infty} \left[ D \left( \iota_{j},\nu+\lambda_{j,j} x \right)-D \left( \iota_{j} , \nu \right) \right] f_j(x) d x, \notag \\
\mu_{k}=&\left( \theta^{\max}_j \right)^{-1} \Bigg\{ \int_{\frac{\xi^{\min}_j \nu}{\lambda_{j,j}}}^{\infty} \big[ E_k \left(\iota_{j}, \nu+ \lambda_{j,j} x \right)- A_k \left( \iota_{j}, \lambda_{j,j} x \right) \notag \\
& - E_k \left( \iota_{j}, \nu \right) \big]  f_j(x) dx + k \mu_{k-1} \Bigg\}, k=1,2,\cdots,
\end{align}
where $\iota_{j}=\lambda_{j,j} x /\xi^{\min}_{j}- \nu$, and the closed-form expressions of $A_k(\cdot,\cdot)$, $D(\cdot,\cdot)$ and $E_k(\cdot,\cdot)$ are given in Appendix B.
\end{lemma}
\begin{proof}
See Appendix B.
\end{proof}
Note that the above results are for general fading scenarios, i.e., interference channels are Nakagami fading, and signal channels are with arbitrary fading. For the special case where all the channels are Rayleigh fading, i.e., Nakagami fading with parameter $1$, closed-form expressions can be obtained, as in \emph{Corollary} 1-2.

\begin{cor}
When the BS can acquire the small-scale fading gain of the communication link $j$, the successful transmission probability for link $j$ using channel $i$ is given as
\begin{align} \label{p1_pr}
p^s_{i,j}
= & 1-
\sum \limits_{z \in \mathcal{L}'_i}
\left( \prod \limits_{k \in \mathcal{L}'_i, k \ne z}
\left( \lambda_{z,j} - \lambda_{k,j}  \right)^{-1} \right)
\lambda_{z,j}^{|\mathcal{L}'_i|-1} \notag \\
 & \times e^{  -\frac{\lambda_{j,j}\beta_{i,j,j}-\xi^\text{min}_j \nu }
{\lambda_{z,j}\xi^\text{min}_j} }, \text{ if } \frac{\lambda_{j,j}\beta_{i,j,j}}{\nu}  \geqslant \xi^{\text{min}}_j,
\end{align}
and zero elsewhere. Furthermore, the expected rate is given as
\begin{align}
r_{i,j}= &- \sum \limits_{z \in \mathcal{L}'_i} \left( \prod \limits_{k \in \mathcal{L}'_i, k \ne z} \left( \lambda_{z,j} - \lambda_{k,j}  \right)^{-1} \right) \lambda_{z,j}^{|\mathcal{L}'_i|-1} \notag \\
& \times \left[ G^1_{z,j}(0) - G^1_{z,j} \left(\frac{\lambda_{j,j}\beta_{i,j,j}}{\lambda_{z,j}} \right) \right] +\log \left(1+\xi^\text{min}_j \right) p^s_{i,j} \notag \\
& + \log \left( \frac{v+\lambda_{j,j}\beta_{i,j,j}}{v+ v \xi^{\min}_j} \right), \text{ if }  \frac{\lambda_{j,j}\beta_{i,j,j}}{\nu }  \geqslant \xi^{\text{min}}_j,
\end{align}
and zero elsewhere, where $G^1_{z,j}(x)=\log(e) \exp \left( \frac{\nu}{\lambda_{z,j}} +x \right)$
$\times \left[ \text{\emph{Ei}} \left( -\frac{\lambda_{j,j}\beta_{i,j,j}}{\xi_j^{\text{min}} \lambda_{z,j}} -x \right)
- \text{\emph{Ei}} \left( -\frac{\nu}{\lambda_{z,j}} -x \right) \right]$.
\end{cor}
\begin{proof}
See Appendix C.
\end{proof}
\begin{cor}
When the BS cannot acquire the small-scale fading gain of communication link $j$, the successful transmission probability and expected rate for link $j$ using channel $i$ is given as
\begin{align} \label{p2_pr}
&p^s_{i,j}=
\exp{\left(- {\xi^\text{min}_j \nu/\lambda_{j,j}}\right)} \notag \\
&\times \left[1-
\sum \limits_{z \in \mathcal{L}'_i}
\left( \prod \limits_{k \in \mathcal{L}'_i, k \ne z} \left( \lambda_{z,j} - \lambda_{k,j}  \right)^{-1} \right)
\frac{\lambda_{z,j}^{|\mathcal{L}'_i|} \xi^\text{min}_j}{\lambda_{j,j} + \lambda_{z,j}\xi^\text{min}_j} \right].
\end{align}
Furthermore, the expected rate is given as
\begin{align}
r_{i,j}=& \log \left(e\right)
\sum \limits_{z \in \mathcal{L}'_i}
\left( \prod \limits_{k \in \mathcal{L}'_i, k \ne z}
\left( \lambda_{z,j} - \lambda_{k,j}  \right)^{-1} \right)
\frac{\lambda_{z,j}^{|\mathcal{L}'_i|-1}}{\lambda_{j,j}-\lambda_{z,j}} \notag \\
& \times \left[ G^2_{z,j}(\lambda_{z,j})-G^2_{z,j}(\lambda_{j,j}) \right] +\log \left(1+\xi^\text{min}_j \right) p^s_{i,j} \notag \\
&-\log \left(e\right)\exp \left( \nu/\lambda_{j,j} \right)  \text{\emph{Ei}} \left[ -{\nu \left(1+ \xi^\text{min}_j \right) }/{\lambda_{j,j}} \right],
\end{align}
where $G^2_{z,j}(x)=\frac{\lambda_{j,j} \lambda_{z,j}}{x} e^{-\nu/x} \text{\emph{Ei}}[-\nu/x-\nu \xi^\text{min}_j/ \lambda_{j,j}]$.
\end{cor}
\begin{proof}
See Appendix D.
\end{proof}

\begin{table*}[!t]
  \renewcommand{\arraystretch}{1.3}
  \caption{Four Scenarios with partial CSI}
  \label{tab_4case}
  \centering
  \begin{tabular}{|c|c|c|c|c|}
  \hline
  &\textbf{Scenario 1}&\textbf{Scenario 2}&\textbf{Scenario 3}&\textbf{Scenario 4}\\
  \hline
  CSI of the cellular communication links & known & known & known & known \\
  \hline
  CSI of the D2D communication links& known & unknown & known & known \\
  \hline
  CSI of the interference links between user devices& unknown & unknown & unknown & unknown\\
  \hline
  CSI of the interference links from the BS to the D2D receivers & known & known & unknown & unknown\\
  \hline
  CSI of the interference links from the D2D transmitters to the BS & known & known & known & unknown\\
  \hline
  \end{tabular}
  \begin{tablenotes}

\item General: 'known' means that the BS can acquire the CSI of this kind of link, and 'unknown' means the BS cannot acquire the instantaneous small-scale fading gains of this kind of link. \end{tablenotes}
\end{table*}

\subsection{Full CSI Scenario}
The problem formulated in (\ref{problem}) can also be applied to the full CSI scenario with the assumption that the BS can acquire the CSI of all links, by setting the successful transmission probability as
\begin{equation}
\text{Pr} \left\{ \xi_{i,j}\left( \mathcal{L}_i \right)  \geqslant \xi^{\text{min}}_{j} \right\}=
\begin{cases}
1 & \text{if } \xi_{i,j} \left( \mathcal{L}_i \right)  \geqslant \xi^{\text{min}}_j.\\
0 & \text{otherwise}.
\end{cases}
\end{equation}
In this simplified case, we can investigate different utility functions, with the following two as examples.
\subsubsection{Weighted Sum-Rate Maximization}
In the full CSI scenario, the expected weighted sum-rate given in (\ref{weighted sr}) can be simplified as the weighted sum-rate given as
\begin{equation} \label{weighted sr}
 U_i \left( \mathcal{L}_i \right)=\sum\limits_{j \in \mathcal{L}_i} w_j \log \left( 1+\xi_{i,j} \left( \mathcal{L}_i \right) \right).
\end{equation}
\subsubsection{Access Rate Maximization}
Most previous works have focused on maximizing the network throughput \cite{limitedarea,Dis,ratep}. However, for different kinds of applications, the users may have different requirements, resulting in different network metrics. For instance, D2D communications can be utilized in machine-to-machine (M2M) communications \cite{survey}, where a large number of low-rate devices need to be supported \cite{M2M5G}. In such scenarios, the \emph{access rate}, which is defined as the ratio of the number of active links to the total links, is more relevant than the network throughput. To maximize the access rate, the utility function of channel $i$ is given as
\begin{equation} \label{weighted sr}
 U_i \left( \mathcal{L}_i \right)=\frac{1}{N}\sum\limits_{j \in \mathcal{L}_i} \mathbbm{1} \left[\xi_{i,j}\left( \mathcal{L}_i \right) \geqslant \xi^{\text{min}}_{j} \right],
\end{equation}
where
\begin{equation} \label{1}
 \mathbbm{1} \left[\xi_{i,j}\left( \mathcal{L}_i \right) \geqslant \xi^{\text{min}}_{j} \right]=
\begin{cases}
1 & \text{if } \xi_{i,j}\left( \mathcal{L}_i \right) \geqslant \xi^{\text{min}}_{j},\\
0 & \text{otherwise}.
\end{cases}
\end{equation}
The access rate maximization problem is equivalent to maximizing the number of users that can be served simultaneously, which has been used as one criterion for \emph{user admission} \cite{userad,S-CRAN}.

Both of these performance metrics can be handled by the algorithms proposed in the following sections, and their performance will be demonstrated in Section VI.

\section{Optimal channel assignment}
The problem formulated in (\ref{problem}) is a mixed integer nonlinear programming (MINLP) problem, which is NP-hard. In this section, we will propose a DP algorithm to find the optimal channel assignment, with a much lower complexity than exhaustive search.

\subsection{Optimal DP Channel Assignment Algorithm}
DP is an efficient technique to find the global optimal solution without requiring differentiability. Thus, it can be applied to deal with non-continuous solution spaces, e.g., integer variables \cite{dpint,dp2}.
In the DP algorithm, the original problem is divided into multiple \emph{stages}, and is solved stage by stage. Each stage is associated with multiple \emph{states}. Meanwhile, there is a recursive relationship connecting the optimal solution for a particular state at one stage and the optimal solutions for previous stages. In other words, the optimal solution at one stage can be constructed by the optima at previous stages. Thus, as long as the optima of the beginning stages are obtained, the optimal solution of the original problem can be constructed according to the recursive relationship. The key design steps in the DP algorithm involve dividing the problem into multiple stages and developing the recursive relationships among the stages, both of which need to be tailored for each specific problem. In the following, we will first identify how to divide the problem formulated in (\ref{problem}) into multiple stages and identify the associated states, and then find the recursive relationship. A simple example is shown in Fig. \ref{DPsample} to illustrate the main idea.

\begin{figure*}[!t]
  \centering
  \includegraphics[width=6.5in]{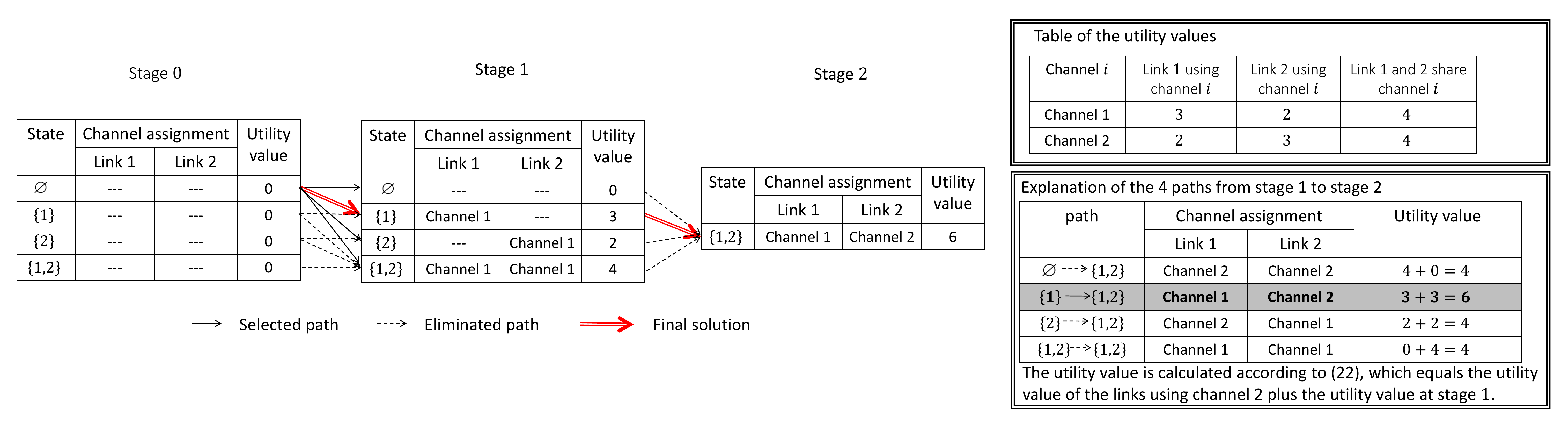}
  \caption{An illustration of the proposed DP algorithm when $M=2$ and $N=2$. In this example, there are three stages, i.e., stages $0$, $1$, and $2$. Stages $0$ and $1$ are with $4$ states, i.e., $\varnothing$, $\{1\}$, $\{2\}$, and $\{1,2\}$, while stage 3 only has one state, i.e., $\{1,2\}$. Stage $0$ is the beginning stage, where all channels have not been assigned. Then, at the $k$-th stage, the $k$-th channel will be assigned by solving problem (\ref{problem_x}). Specifically, in the final solution, the first transition is from state $\varnothing$ to state $\{1\}$, with channel $1$ assigned to link $1$, and the second transition is from state $\{1\}$ to state $\{1,2\}$, with channel $2$ assigned to link $2$.}
  \label{DPsample}
\end{figure*}
\subsubsection{Stages and States}
In our proposed DP algorithm, each stage corresponds to a particular channel, and thus the stage by stage processing of the algorithm is to assign channels one by one. Meanwhile, the states represent different subsets of $\mathcal{S}$, and each state contains a particular group of links. As shown in Fig. \ref{DPsample}, the final result of the algorithm is a particular path selected from the state transition diagram.

To specify the algorithm, we first define a few notations. Denote $\mathcal{CH}^{k}  \triangleq \{i|1 \leqslant i \leqslant k \}$. At the $k$-th stage, $0 \leqslant k \leqslant M$, for a particular set of links $\mathcal{J} \subseteq \mathcal{S}$, we need to find the optimal channel assignment for links in $\mathcal{J}$ sharing the $k$ channels in $\mathcal{CH}^{k}$. Based on the recursive relationship, as will be specified in Section IV-A-2), each state at the $k$-th stage can be transferred from the states at the $(k-1)$-th stage, while the one with the maximum utility value is selected and others are eliminated. Thus, the problem can be solved stage by stage, and the complexity can be significantly reduced thanks to the recursion among consecutive stages. Finally, at the last stage, the optimal channel assignment of the original problem (\ref{problem}) can be found.

Next, we will find the utility value at the $k$-th stage associated with the state $\mathcal{J}$, denoted by $OPT_{k,\mathcal{J}}$, and $OPT_{M, \mathcal{S}}$ will give the optimal utility value of the original problem. Finding $OPT_{k,\mathcal{J}}$ is to maximize the utility of links in $\mathcal{J}$ sharing $k$ channels in $\mathcal{CH}^{k}$ while guaranteeing the constraints (\ref{problem}a)-(\ref{problem}f), which is equivalent to solving the following optimization problem:
\begin{align}
\mathop {\max } \limits_{\mathcal{L}_i \in \mathcal{J}, i \in \mathcal{CH}^k} & \sum \limits_{i \in \mathcal{CH}^k} U_i \left( \mathcal{L}_i \right)= U_k(\mathcal{L}_k) + \sum \limits_{i \in \mathcal{CH}^{k-1}} U_i \left( \mathcal{L}_i \right), \label{problem_s} \\
\text{s.t. } & \text{Pr} \left\{ \xi_{i,j}\left( \mathcal{L}_i \right) \geqslant \xi^{\text{min}}_{j} \right\}  \geqslant \psi^{\text{min}}_{j}, \notag \\
& \forall i \in \mathcal{CH}^k \text{ and } \forall j \in \mathcal{L}_i, \tag{\ref{problem_s}a} \\
& \left| \mathcal{C} \cap \mathcal{L}_i \right| \leqslant 1, \forall i \in \mathcal{CH}^k, \tag{\ref{problem_s}b} \\
& \sum \limits_{i \in \mathcal{CH}^k} |\mathcal{L}_i \cap \mathcal{C}| = |\mathcal{C} \cap \mathcal{J}|, \tag{\ref{problem_s}c} \\
& \mathcal{L}_{u} \cap \mathcal{L}_{v} = \varnothing, \notag \\
& \forall u \ne v \text{ and } u \in \mathcal{CH}^k \text{ and } v \in \mathcal{CH}^k, \tag{\ref{problem_s}d} \\
& \mathcal{L}_i \cap \mathcal{C}_u= \varnothing, \forall i \in \mathcal{CH}_d \cap \mathcal{CH}^k, \tag{\ref{problem_s}e} \\
& \mathcal{L}_i \cap \mathcal{C}_d= \varnothing, \forall i \in \mathcal{CH}_u \cap \mathcal{CH}^k, \tag{\ref{problem_s}f}
\end{align}
where $|\cdot |$ denotes the cardinality, $\mathcal{L}_i$ denotes the set of links using channel $i$, constraint (\ref{problem_s}a) guarantees the QoS requirements, constraint (\ref{problem_s}b) implies each channel cannot be accessed by more than one cellular link, constraints (\ref{problem_s}c) and (\ref{problem_s}d) represent the constraints (\ref{problem}c) and (\ref{problem}d), and constraints (\ref{problem_s}e) and (\ref{problem_s}f) means that uplink cellular links can only access the uplink channel, so as for the downlink cellular links. Denote the constraint set of this problem as $CONS_{k,\mathcal{J}}$, and the optimal channel assignment of link $j \in \mathcal{S}$ is denoted as $\left(\rho_{k,\mathcal{J}} \right)_{i,j}$, which is defined similarly to (\ref{rho}). This problem is computationally difficult, and in the following we will develop recursive relationship to find the optimal solution.

\subsubsection{Recursive Relationship}
At the $0$-th stage, no channel is assigned yet. Hence, the optimum value $OPT_{0, \mathcal{J}}$, $\mathcal{J} \subseteq \mathcal{S}$, equals $0$. To find the recursive relationship, we need to transfer the problem of finding the optimum value $OPT_{k,\mathcal{J}}$ to the problem of finding the optimum values at the $(k-1)$-th stage.

Suppose that $\mathcal{L}_k$ is fixed, and then the channels in $\mathcal{CH}^{k-1}$ can only be assigned to the links in $\left( \mathcal{J}-\mathcal{L}_k \right)$ considering the constraint (\ref{problem_s}d). In the following, we will show that solving the problem at the $k$-th stage associated with state $\mathcal{J}$ can be transferred to solving the problem at the $(k-1)$-th stage associated with state $(\mathcal{J}-\mathcal{L}_k)$. Firstly, the constraints should keep consistent. Specifically, for the sets $\mathcal{L}_{i}$, $1 \leqslant i \leqslant k-1$, the constraint set $CONS_{k,\mathcal{J}}$ should be equivalent to $CONS_{k-1,\mathcal{J}-\mathcal{L}_k}$. It is obvious for constraints (\ref{problem_s}a), (\ref{problem_s}b), and (\ref{problem_s}d)-(\ref{problem_s}f), since they are irrelevant of the state $\mathcal{J}$. For constraint (\ref{problem_s}c), it can be rewritten as
$\sum \limits_{i \in \mathcal{CH}^k} |\mathcal{L}_i \cap \mathcal{C}| = |\mathcal{C} \cap \mathcal{J}| \Rightarrow \sum \limits_{i \in \mathcal{CH}^{k-1}} |\mathcal{L}_i \cap \mathcal{C}| + |\mathcal{L}_k \cap \mathcal{C}| = |\mathcal{C} \cap (\mathcal{J}-\mathcal{L}_k)| + |\mathcal{L}_k \cap \mathcal{C}| \Rightarrow \sum \limits_{i \in \mathcal{CH}^{k-1}} |\mathcal{L}_i \cap \mathcal{C}| = |\mathcal{C} \cap (\mathcal{J}-\mathcal{L}_k)|.
$ Secondly, the objective value $OPT_{k,\mathcal{J}}$ with fixed $\mathcal{L}_k$, denoted as $OPT_{k,\mathcal{J}}^{\mathcal{L}_k}$, can be calculated based on $OPT_{k-1,\mathcal{J}-\mathcal{L}_k}$, given as
\begin{align}
OPT_{k,\mathcal{J}}^{\mathcal{L}_k}&=\mathop {\max } \limits_{\mathcal{L}_i, i \in \mathcal{CH}^{k-1}} \bigg[ U_k(\mathcal{L}_k) + \sum \limits_{i \in \mathcal{CH}^{k-1}} U_i \left( \mathcal{L}_i \right) \bigg] \notag \\
& = U_k(\mathcal{L}_k) + \mathop {\max } \limits_{\mathcal{L}_i, i \in \mathcal{CH}^{k-1}} \sum \limits_{i \in \mathcal{CH}^{k-1}} U_i \left( \mathcal{L}_i \right) \notag \\
&=  U_k \left( \mathcal{L}_k \right) + OPT_{k-1, \mathcal{J} - \mathcal{L}_k}
\end{align}

After searching for all the feasible selections for the set $\mathcal{L}_k$, in which all the links can share channel $k$ simultaneously, the optimum value $OPT_{k,\mathcal{J}}$, $1 \leqslant k \leqslant M$, can be found according to the following recursive relationship:
\begin{align}
OPT_{k,\mathcal{J}} = \mathop {\max } \limits_{\mathcal{L}_k } & \left[ U_k \left( \mathcal{L}_k \right) + OPT_{k-1, \mathcal{J} - \mathcal{L}_k} \right] , \label{problem_x}\\
\text{s.t. } & \text{Pr} \left\{ \xi_{k,j}\left( \mathcal{L}_k \right) \geqslant \xi^{\text{min}}_{j} \right\}  \geqslant \psi^{\text{min}}_{j}, \forall j \in \mathcal{L}_k, \tag{\ref{problem_x}a} \\
& \left| \mathcal{C} \cap \mathcal{L}_k \right| \leqslant 1, \tag{\ref{problem_x}b} 
\end{align}
\begin{align}
& \left| \mathcal{C} \cap \mathcal{L}_k \right| = 1,  \text{if } k \in \mathcal{CH}_x \notag \\
& \text{ and } \left| \mathcal{CH}_x \cap \mathcal{CH}^k \right| \leqslant \left| \mathcal{C}_x \cap \mathcal{J} \right|, x \in \{u,d\}, \tag{\ref{problem_x}c} \\
& \mathcal{L}_k \cap \mathcal{C}_x = \varnothing, \text{if } k \notin \mathcal{CH}_x, x \in \{u,d\}, \tag{\ref{problem_x}d}
\end{align}
where constraints (\ref{problem_x}a)-(\ref{problem_x}c) ensure constraints (\ref{problem_s}a)-(\ref{problem_s}c), respectively, and constraint (\ref{problem_x}d) implies constraints (\ref{problem_s}e) and (\ref{problem_s}f).

\subsubsection{Optimal Channel Assignment}
Denote $\mathcal{L}_k^\star = \mathop {\arg\max } \limits_{\mathcal{L}_k} \left[ U_k \left( \mathcal{L}_k \right) + OPT_{k-1, \mathcal{J} - \mathcal{L}_k} \right]$, and then the optimal channel assignment can be identified as follows. The links in $\mathcal{L}_k^\star$ will use channel $k$, while the links using channel $i$, for $1 \leqslant i < k$, are the same as those at the $(k-1)$-th stage associated with state $(\mathcal{J}-\mathcal{L}_k^\star)$. Thus, the optimal channel assignment at the $k$-th stage, $1 \leqslant k \leqslant M$, associated with state $\mathcal{J} \subseteq \mathcal{S}$ is specified by
\begin{equation} \label{rho_cdp}
\left(\rho_{k,\mathcal{J}} \right)_{i,j}=\left\{
                \begin{array}{ll}
                \left(\rho_{k-1,\mathcal{J}-\mathcal{L}_k^\star} \right)_{i,j} & \text{if } i < k,\\
                1& \text{if } i = k \text{ and link } j \in \mathcal{L}^\star_k, \\
                0& \text{if } i = k \text{ and link } j \notin \mathcal{L}_k^\star, \\
                \text{undecided} & \text{if } i > k,
                \end{array}
                \right.
\end{equation}
where $i \in \mathcal{CH}, j \in \mathcal{S}$.
No channel is assigned at the $0$-th stage, and when $k >0 $, the optimal solution at the $k$-th stage can be found based on the optimal solution at the $(k-1)$-th stage.
Therefore, we can find the optimal solution at the $M$-th stage by recursion, which is also optimal for the original channel assignment problem (\ref{problem}). To sum up, the DP algorithm is described in Algorithm \ref{alg:dp}.

The proposed DP algorithm provides a relatively efficient approach, compared to exhaustive search, to get the optimal solution. Finding the optimal solution is important for two reasons in this paper: 1) it provides a performance upper bound for evaluating sub-optimal algorithms; 2) it is needed to compare different CSI scenarios to reveal the relative importance of the CSI of different links, for which sub-optimal algorithms may lead to illusive conclusions.
\begin{algorithm}[!t]
\caption{The optimal DP algorithm}
\label{alg:dp}
\begin{algorithmic}[1]
\STATE $OPT_{M,\mathcal{S}}, \rho_{M,\mathcal{S}}$: the optimal utility value and  channel assignment for the original problem
\FORALL{$\mathcal{J} \subseteq \mathcal{S}$}
\STATE {Set $OPT_{0,\mathcal{J}}=0$.}
\ENDFOR
\FOR{$k=1$ to $M-1$}
\FORALL{$\mathcal{J} \subseteq \mathcal{S}$}
\STATE {Set $OPT_{k,\mathcal{J}}=0$.}
\FORALL{feasible $\mathcal{L}_k$ that satisfies constraints in (\ref{problem_x})}
\IF{$OPT_{k,\mathcal{J}} <  U_k \left( \mathcal{L}_k \right) + OPT_{k-1, \mathcal{J} - \mathcal{L}_k}$}
\STATE Set $OPT_{k,\mathcal{J}} =  U_k \left( \mathcal{L}_k \right) + OPT_{k-1, \mathcal{J} - \mathcal{L}_k}$ and $\mathcal{L}_k^\star = \mathcal{L}_k$.
\ENDIF
\ENDFOR
\STATE Update $\left(\rho_{k,\mathcal{J}} \right)_{i,j}$, $i \in \mathcal{CH}$, $j \in \mathcal{S}$, according to (\ref{rho_cdp}).
\ENDFOR
\ENDFOR
\STATE {Set $k=M$,$\mathcal{J}=\mathcal{S}$, and Run Steps 7-13.}
\end{algorithmic}
\end{algorithm}

\subsection{Complexity Analysis}
Based on Algorithm \ref{alg:dp}, the time complexity to obtain $OPT_{k,\mathcal{J}}$ is given as
\begin{equation} \label{tc_dp1}
\begin{aligned}
& \mathcal{T}_{state} \left( k,\mathcal{J} \right) = \mathcal{O} \left( \left| \mathcal{S}^p_{k,\mathcal{J}} \right| C \right) \\
&=\left\{ 
\begin{array}{ll}
\mathcal{O} \left( (\left| \mathcal{C}_u \cap \mathcal{J} \right|+1) \cdot 2^{\left| \mathcal{D} \cap \mathcal{J} \right|} C \right)& \text{if } R_k \in \mathcal{R}_u, \\
\mathcal{O} \left( (\left| \mathcal{C}_d \cap \mathcal{J} \right|+1) \cdot 2^{\left| \mathcal{D} \cap \mathcal{J} \right|} C \right)& \text{if } R_k \in \mathcal{R}_d,
\end{array}
\right.
\end{aligned}
\end{equation}
where $C$ is the time complexity to check constraint (\ref{problem_x}a) and to calculate the utility function $U_k \left( \mathcal{L}_k \right)$.
The time complexity of the $k$-th stage when $k<M$ can then be derived as
\begin{equation} \label{tc_dp2}
\mathcal{T}_{stage}  \left( k \right)
=\left\{
\begin{array}{ll}
\mathcal{O} \left((N_{uc}/2+1) 2^{N_{uc}} 3^{N_d} C \right)& \text{if } R_k \in \mathcal{R}_u, \\
\mathcal{O} \left((N_{dc}/2+1) 2^{N_{dc}} 3^{N_d} C \right)& \text{if } R_k \in \mathcal{R}_d.
\end{array}
\right.
\end{equation}
For the $M$-th stage, we only need to consider the state $\mathcal{S}$. Thus, the time complexity of the $M$-th stage can be expressed as
\begin{equation} \label{tc_dp3}
\mathcal{T}_{stage}  \left( M \right)
=\left\{
\begin{array}{ll}
\mathcal{O} \left( (N_{uc}+1) \cdot 2^{N_d} C \right)& \text{if } M_d=0,\\
\mathcal{O} \left( (N_{dc}+1) \cdot 2^{N_d} C \right)& \text{if } M_d>0.
\end{array}
\right.
\end{equation}
Based on (\ref{tc_dp2}) and (\ref{tc_dp3}), and by ignoring the constant coefficients and lower order terms, the upper bound for the time complexity of our proposed DP algorithm can be derived as
\begin{equation} \label{tc_dpall}
\begin{aligned}
&\mathcal{T}_{optimal} = \\
&
\begin{cases}
\mathcal{O} \Big( M 2^{N_d} C \Big) & \text{if } M \leqslant 1,\\
\mathcal{O} \Big( \left(M_{u} N_{uc} 2^{N_{uc}} + M_{d} N_{dc} 2^{N_{dc}}\right)  3^{N_d} C \Big) & \text{otherwise}.
\end{cases}
\end{aligned}
\end{equation}
By comparing the time complexity of the exhaustive search method, which is
\begin{equation} \label{tc_ex}
\mathcal{T}_{search} = \mathcal{O} \left( \frac{M_{u}!}{\left( M_u-N_{uc} \right)!} \cdot \frac{M_{d}!}{\left( M_d-N_{dc} \right)!} (M+1)^{N_d} C \right),
\end{equation}
we can find that our proposed optimal algorithm performs similarly to exhaustive search when the number of channels $M \leqslant 2$ and performs much more efficiently when $M>2$. Although the proposed algorithm has an exponential worst-case complexity, it can serve as a performance benchmark to evaluate other algorithms. One limitation of the DP algorithm is that it requires a large memory space, which is $\mathcal{O} \Big( M 2^{\max \left( N_{uc}, N_{dc}\right)} 2^{N_d} \left(N_c+N_d\right)\Big)$.

\section{cluster-based sub-optimal channel assignment}
\begin{figure}[!t]
  \centering
  \includegraphics[width=3.5in]{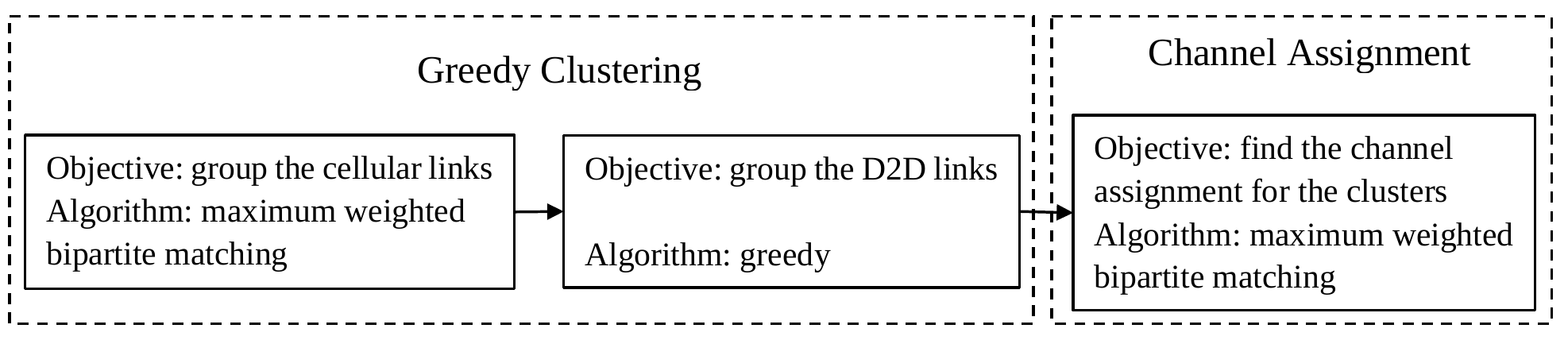}
  \caption{Cluster-based sub-optimal channel assignment algorithm.}
  \label{sub-optimal}
\end{figure}
In this section, a low-complexity sub-optimal algorithm is proposed, which includes two steps. As shown in Fig. \ref{sub-optimal}, the first step is to divide all the links into $M$ non-overlapping clusters, denoted as $\mathcal{GC} = \{\mathcal{G}_1,\mathcal{G}_2,...,\mathcal{G}_M\}$, and the second step is to assign the channels to the $M$ clusters. Note that, in the first step, we will temporarily assume that cluster $g$ uses channel $g$, where $1 \leqslant g \leqslant M$. Since such channel assignment may not be optimal, we will then reassign channels to achieve a higher network utility.
\subsection{Greedy Clustering Algorithm}
Since links in the same cluster will share the same channel, we should put those that will generate relatively low mutual interference to the same cluster. Conventional clustering algorithms \cite{kmean,cluster,moretti2013efficient} are inapplicable to our problem, as there are two types of links and the cellular links have a higher priority than D2D links. Furthermore, the utility function is highly complex. Thus, we will propose different methods to cluster cellular and D2D links, respectively. In particular, a bipartite matching algorithm is proposed to cluster the cellular links, which also helps to check the feasibility of the original channel assignment problem. To cluster D2D links, we propose a greedy algorithm, for which different priority values are defined for different utility functions based on the key features of the considered system. The overall greedy clustering algorithm is shown in Algorithm \ref{alg:cluster}, where $\mathcal{U}=\mathcal{S} - \bigcup_{g=1}^{M} \mathcal{G}_g$ denotes the set of links that have not been put to any cluster, and $\mathcal{Q}_g$, $1 < g \leqslant M$, is a First-In-First-Out (FIFO) queue representing the sequence of links adding to cluster $\mathcal{G}_g$, where the link being put first will be processed first. The FIFO queue will also be important for the channel assignment in the second step. In the following, we will explain how the greedy clustering algorithm works.

\subsubsection{Cellular Link Clustering}
We first build a bipartite graph by regarding the $M$ clusters and $N_c$ cellular links as two groups of vertexes. Then, we will define the weights of the edges. Considering the constraints (\ref{problem}a), (\ref{problem}e) and (\ref{problem}f), while putting a cellular link $j \in \mathcal{C}$ into cluster $\mathcal{G}_g \in \mathcal{GC}$, the following conditions should be satisfied:
\begin{enumerate}[a)]
  \item ${\frac{p^t_j h_{g,j} }{\sigma^2 }} \geqslant \xi^{\text{min}}_{j}$,\label{cu_a}
  \item If $j \in \mathcal{C}_u$, $g$ can only be selected from $[1,M_u]$; otherwise, $g$ is selected from $(M_u,M]$, \label{cu_b}
\end{enumerate}
where condition \ref{cu_a}) guarantees the QoS constraints, and condition \ref{cu_b}) addresses constraints (\ref{problem}e) and (\ref{problem}f). Thus, the weight of the edge between cluster $\mathcal{G}_g$ and cellular link $j$ is defined as
\begin{equation}
T^c_{g,j}=
\begin{cases}
r_{g,j}& \text{if conditions \ref{cu_a}) and \ref{cu_b}) are satisfied}, \\
-\infty & \text{ otherwise},
\end{cases}
\end{equation}
where $r_{g,j}$ is different for different utility functions. For instance, we set $r_{g,j} = a_j \log \left( 1 + \frac{ p^t_j h_{g,j} }{\sigma^2 } \right)$ for the weighted sum-rate, and $r_{g,j}=\log \left( 1 + \frac{ p^t_j h_{g,j} }{\sigma^2 } \right)$ for the access rate.

According to constraints (\ref{problem}b) and (\ref{problem}c), different cellular links should be assigned to different clusters. Thus, cellular links are put into distinct clusters by maximizing the weighted sum-rate, which is to solve the following maximum weighted bipartite matching problem
\begin{align}
\mathop {\max }\limits_{\beta^c_{i,j} } &{ {\sum\limits_{g: \mathcal{G}_g \in \mathcal{GC},j \in \mathcal{C}} {\beta^c_{g,j}T^c_{g,j}  } } }, \label{problem_c} \\
\text{s.t. } &\sum\limits_{j \in \mathcal{C}} {\beta^c _{g,j} } \leqslant 1, \beta^c_{g,j} \in \{0,1\}, \forall g: \mathcal{G}_g \in \mathcal{GC},\tag{\ref{problem_c}a} \\
&\sum\limits_{g: \mathcal{G}_g \in \mathcal{GC}} {\beta^c _{g,j} } = 1, \beta^c_{g,j} \in \{0,1\}, \forall j \in \mathcal{C},\tag{\ref{problem_c}b}
\end{align}
where $\beta^c _{g,j}=1$ indicates that cellular link $j$ is put into cluster $\mathcal{G}_g$. This problem can be solved efficiently by the Kuhn-Munkres (KM) algorithm \cite{KM}. Moreover, such an algorithm can also check the feasibility of the problem formulated in (\ref{problem}). If the optimal solution of problem (\ref{problem_c}) is non-negative, each cellular link can access one channel and meet the QoS requirement. Accordingly, a feasible channel assignment for problem (\ref{problem}) is found. Otherwise, the QoS requirements for the cellular links cannot be satisfied simultaneously.
\subsubsection{D2D Link Clustering}
After clustering the cellular links, we need to cluster the D2D links. The proposed greedy algorithm is shown in lines 9 to 14 in Algorithm \ref{alg:cluster}. For this greedy algorithm to work, we will first define a priority value for each cluster-link pair $(g,j)$, denoted as $y_{g,j}$. A higher priority value $y_{g,j}$ implies that link $j$ has a higher priority to be put into cluster $\mathcal{G}_g$. In each greedy step, a cluster-link pair $\left(g^\star,j^\star\right)$ with the highest priority value is selected, and the corresponding D2D link $j^\star$ will be put into cluster $\mathcal{G}_{g^\star}$. Thus, the priority values should be chosen prudently for different utility functions. In the following, we will show some examples of priority values.
\begin{enumerate}[a)]
\item{\textbf{Weighted Sum-Rate Maximization}:} With the weighted sum-rate as the objective, we first define the utility gain of putting D2D link $j$ into cluster $\mathcal{G}_{g}$ as $v_{g,j}=U_g \left( \mathcal{G}_{g} \cup \{j\} \right)-U_g \left( \mathcal{G}_{g} \right)$. Intuitively, a larger utility gain will result in a higher priority value. Thus, we define the priority value of a D2D link $j \in \mathcal{U}$ being put into cluster $\mathcal{G}_{g}$ as
\begin{align} \label{operator}
&y_{g,j}= \notag \\
&\left\{
\begin{array}{ll}
\multirow{2}{*}{$v_{g,j}$}
& \text{if } \text{Pr} \left\{ \xi_{g,z}\left( \mathcal{G}_{g} \cup \{j\} \right) \geqslant \xi^{\text{min}}_{z} \right\}  \geqslant \psi^{\text{min}}_{z}, \\
 & \forall z \in \mathcal{G}_{g} \cup \{j\}, \\
\multirow{2}{*}{$v_{g,j}$} & \text{if } \prod \limits_{z \in \mathcal{G}_{g} \cup \{j'\}}  \mathbbm{1} \Big[ \text{Pr} \left\{ \xi_{g,z}\left( \mathcal{G}_{g} \cup \{j'\} \right) \geqslant \xi^{\text{min}}_{z} \right\} \\
&  \geqslant \psi^{\text{min}}_{z}  \Big] =0, \forall j' \in \mathcal{U} \text{ and } 1 \leqslant g \leqslant M ,\\
-\infty & \text{otherwise}.
\end{array}
\right.
\end{align}
The first condition means QoS constraint (\ref{problem}a) is satisfied if link $j$ is put into cluster $\mathcal{G}_g$, and thus the priority value is just its utility gain. The second condition means that for every link in $\mathcal{U}$, it will violate QoS constraint (\ref{problem}a) if being put into any cluster. For this special case, the priority value is also set as the utility gain, which is to make sure every link will be clustered. For the third condition, there exist some links in $\mathcal{U}$ that will violate the QoS constraint if being clustered, and thus such links will get the lowest priority values, i.e., $-\infty$, to allow other good links to be clustered.
\item{\textbf{Access Rate Maximization}:} With this utility, the definition of the priority value should consider two aspects. The first aspect is the minimum ratio between the link rate of the currently considered link and the minimum link rate requirement among all the links in a cluster. A larger minimum ratio results in a higher priority value. The other aspect is the quantity $f^l_{j} \triangleq \sum \limits_{1 \leqslant g \leqslant M} \prod \limits_{z \in \mathcal{G}_{g} \cup \{j\}} \mathbbm{1} \left[ \xi_{g,z}\left( \mathcal{G}_{g} \cup \{j\}\right) \geqslant \xi^{\text{min}}_{z}  \right]$, which denotes the number of clusters included in $\{\mathcal{G}_{g} \cup \{j\} | 1 \leqslant g \leqslant M\}$, in which all links satisfy their QoS requirements. A larger $f^l_{j}$ means that more clusters are able to serve link $j$, which means that there are more choices for link $j$, and thus, we set a lower priority value.
Thus, we define the priority value $y_{g,j}$ to maximize the access rate as
\begin{equation} \label{operator}
y_{g,j}=
\begin{cases}
\min \limits_{ z \in \mathcal{G}_{g} \cup \{j\}} \frac{\log \left[ 1+ \xi_{g,z}\left( \mathcal{G}_{g} \cup \{j\}\right) \right]2^{-f^l_{j}}}{\log \left( 1 + \xi^{\text{min}}_{z} \right)} & \text{if } f^l_{j}>0 ,\\
\min \limits_{ z \in \mathcal{G}_{g} \cup \{j\}} \frac{ \log \left[ 1+ \xi_{g,z}\left( \mathcal{G}_{g} \cup \{j\}\right) \right]2^{-M}}{\log \left( 1 + \xi^{\text{min}}_{z} \right)} & \text{otherwise}.
\end{cases}
\end{equation}
\end{enumerate}

\begin{algorithm}[!t]
\caption{Greedy clustering algorithm}
\label{alg:cluster}
\begin{algorithmic}[1]
\STATE {$\mathcal{U}$: the set of links that have not been put to any cluster}
\STATE {$\mathcal{Q}_g$: a queue representing the sequence of links adding to cluster $\mathcal{G}_g$}
\STATE {$y_{g,j}$: the priority value of putting link $j$ to cluster $\mathcal{G}_g$}
\STATE {Set $\mathcal{G}_g=\varnothing$, $1 \leqslant g \leqslant M$.}
\STATE {Put the cellular links into the $M$ clusters by solving problem (\ref{problem_c}) using the KM algorithm and update $\mathcal{G}_g$, $1 \leqslant g \leqslant M$.}
\FOR {$g=1$ to $M$}
\STATE {Add cellular link in $\mathcal{G}_g$ to the queue $\mathcal{Q}_g$.}
\ENDFOR
\STATE {Update the priority values $y_{g,j}$, $1 \leqslant g \leqslant M$, $j \in \mathcal{U}$.}
\WHILE{$\mathcal{U} \ne \varnothing$}
\STATE {Select $(g^\star,j^\star)=\arg\max \limits_{g: \mathcal{G}_g \in \mathcal{GC},j \in \mathcal{U}} y_{g,j}$.}
\STATE {Put D2D link $j^\star$ into cluster $\mathcal{G}_{g^\star}$: $\mathcal{G}_{g^\star}=\mathcal{G}_{g^\star} \cup \{j^\star\}$, and add D2D link $j^\star$ to queue $\mathcal{Q}_{g^\star}$.}
\STATE {Update the priority values $y_{g^\star,j}$, $j \in \mathcal{U}$.}
\ENDWHILE
\end{algorithmic}
\end{algorithm}

\subsection{Sub-Optimal Channel Assignment Algorithm}
In the first step, we have divided all the cellular and D2D links into $M$ non-overlapping clusters assuming that cluster $\mathcal{G}_g$ will be assigned with channel $g$. However, such channel assignment may not be optimal. Thus, in the second step, we will consider improving the channel assignment for clusters based on maximum weighted bipartite matching. In particular, a bipartite graph is built by regarding the $M$ channels and $M$ clusters as two groups of vertexes. We will first define the weights of the edges, and then solve the maximum weighted bipartite matching problem to get the channel assignment. The main difficulty in this step is to define the weights, where a low-complexity sub-optimal algorithm, i.e., Algorithm \ref{alg:t}, is proposed, built upon a unique queue structure developed during the clustering step.

Firstly, we define the weight of the edge between channel $i$ and cluster $g$ as the maximum utility value of the links in cluster $\mathcal{G}_g$ using channel $i$, denoted as $T^w_{i,g}$. To get the weight value $T^w_{i,g}$, a set of links $\mathcal{SG}_{i,g} \subseteq \mathcal{G}_g$ should be selected to share channel $i$, which maximizes the utility value. Then, the problem can be formulated as
\begin{align}
T^w_{i,g}=\mathop {\max }\limits_{ \mathcal{SG}_{i,g} }& \quad U_{i}(\mathcal{SG}_{i,g}),  \label{problem_d} \\
\text{s.t. }&\text{Pr} \left\{ \xi_{i,j}\left( \mathcal{SG}_{i,g} \right) \geqslant \xi^{\text{min}}_{j} \right\}  \geqslant \psi^{\text{min}}_{j}, \forall j \in \mathcal{SG}_{i,g},\tag{\ref{problem_d}a}\\
&\mathcal{SG}_{i,g} \subseteq \mathcal{G}_g,\tag{\ref{problem_d}b} \\
&\mathcal{SG}_{i,g} \cap \mathcal{C} = \mathcal{G}_g \cap \mathcal{C},\tag{\ref{problem_d}c} \\
&j \notin \mathcal{SG}_{i,g}, \text{ if } i \in \mathcal{CH}_x \text{ and } j \notin \mathcal{C}_x, x \in \{u,d\}, \tag{\ref{problem_d}d}
\end{align}
where constraint (\ref{problem_d}a) implies the QoS requirements, constraint (\ref{problem_d}c) guarantees the priority of cellular links, and constraints (\ref{problem_d}d) represents the constraints (\ref{problem}e) and (\ref{problem}f) of the original problem. If the problem is infeasible, we set $T^w_{i,g}$ as $-\infty$. The problem formulated in (\ref{problem_d}) is a mixed integer programming problem, which is NP hard. Thus, we propose a sub-optimal queueing-based algorithm, as shown in Algorithm \ref{alg:t}, which utilizes the FIFO queues developed in the clustering algorithm, i.e., $\mathcal{Q}_g$, to significantly reduce the complexity. In Algorithm \ref{alg:t}, a cellular link in $\mathcal{G}_g$ is first added to $\mathcal{SG}_{i,g}$. Then, the D2D links in $\mathcal{G}_g$, which can guarantee the QoS requirements (\ref{problem_d}a), are added to $\mathcal{SG}_{i,g}$ according to the sequence in queue $\mathcal{Q}_g$, and a sequence of sets $\mathcal{SG}^{(t)}_{i,g}$, $0 \leqslant t \leqslant \left| \mathcal{G}_g \right|$, is generated. Finally, the optimal $\mathcal{SG^\star}_{i,g}$ is selected among the sets $\mathcal{SG}^{(t)}_{i,g}$, $0 \leqslant t \leqslant \left| \mathcal{G}_g \right|$.
\begin{algorithm}[!t]
\caption{Queuing-based algorithm to obtain $T^w_{i,g}$.}
\label{alg:t}
\begin{algorithmic}[1]
\STATE{$\mathcal{SG^\star}_{i,g}$: the optimal set of links selected from $\mathcal{G}_g$ sharing channel $i$}
\STATE{Set $t=0$ and $\mathcal{SG}^{(0)}_{i,g} = \mathcal{G}_g \cap \mathcal{C}$.}
\IF {$\mathcal{SG}^{(0)}_{i,g}$ satisfies the constraints (\ref{problem_d}a), (\ref{problem_d}d) and (\ref{problem_d}e).}
\FOR {Pick one D2D link $j$ from queue $\mathcal{Q}_g$ in sequence.}
\IF {$\mathcal{SG}^{(t)}_{i,g} \cup \{j\}$ satisfies constraint (\ref{problem_d}a). }
\STATE{Set $t=t+1$ and $\mathcal{SG}^{(t)}_{i,g} = \mathcal{SG}^{(t-1)}_{i,g}  \cup \{j\}$.}
\ENDIF
\ENDFOR
\STATE{Select $b^\star = \arg\max \limits_{b \in \{0,1...,t\}}U_{i} \left( \mathcal{SG}^{(b)}_{i,g} \right)$.}
\STATE{Set $\mathcal{SG^\star}_{i,g}=\mathcal{SG}^{(b^\star)}_{i,g} $.}
\STATE{Calculate $T^w_{i,g}= U_{i} \left( \mathcal{SG^\star}_{i,g} \right)$.}
\ELSE
\STATE{Set $T^w_{i,g}=-\infty$.}
\ENDIF
\end{algorithmic}
\end{algorithm}

Secondly, a bipartite graph is built by regarding the $M$ channels and $M$ clusters as two groups of vertexes and $T^w_{i,g}$ as the weight of the edge between channel $i$ and cluster $g$. In order to maximize the overall utility value, we formulate the channel assignment problem for clusters as the maximum weighted bipartite matching problem \cite{kim2006use}, given by
\begin{align}
\mathop {\max }\limits_{\beta^w_{i,g} } &{ {\sum\limits_{i \in \mathcal{CH},g: \mathcal{G}_g \in \mathcal{GC}} {\beta^w_{i,g}T^w_{i,g}  } }, }\label{problem_sub} \\
\text{s.t. }&\sum\limits_{i \in \mathcal{CH}} {\beta^w _{i,g} } \leqslant 1, \beta^w_{i,g} \in \{0,1\}, \forall g: \mathcal{G}_g \in \mathcal{GC},\tag{\ref{problem_sub}a} \\
&\sum\limits_{g: \mathcal{G}_g \in \mathcal{GC}} {\beta^w _{i,g} } \leqslant 1, \beta^w_{i,g} \in \{0,1\}, \forall i \in \mathcal{CH},\tag{\ref{problem_sub}b}
\end{align}
where $\beta^{w} _{i,g}=1$ denotes that links in cluster $g$ access channel $i$. The optimal solution of the problem formulated in (\ref{problem_sub}), denoted as $\beta^{w\star} _{i,g}$, can be found by the KM algorithm.

Finally, when $\beta^{w\star} _{i,g}=1$, active links in cluster $\mathcal{G}_g$ will access channel $i$. Thus, the sub-optimal channel assignment for both cellular and D2D links is given by
\begin{equation}
\rho_{i,j}=\left\{
                \begin{array}{ll}
                1 & \text{ if } \beta^{w\star} _{i,g}=1 \text{ and } j \in \mathcal{SG}^\star_{i,g}, g: j \in \mathcal{G}_{g},\\
                0& \text{otherwise},
                \end{array}
                \right.
\end{equation}
where $i \in \mathcal{CH}$ and $j \in \mathcal{S}$.

\subsection{Complexity Analysis}
According to Algorithm \ref{alg:t}, different results of the clustering algorithm will lead to different time complexities. Therefore, we only show the time complexity of the worst case as $\mathcal{T}_{suboptimal}=\mathcal{O} \left( M^3 + N_d^2C +MN_dC \right)$, where $C$ is the time complexity to check the successful transmission probability requirements and to calculate the utility function, and the constant coefficients and lower order terms are ignored.

\section{simulation results}

\begin{table}[!t]
  \renewcommand{\arraystretch}{1.3}
  \caption{Simulation Parameters}
  \label{tab1}
  \centering
  \begin{tabular}{|c|c|}
  \hline
  \textbf{Parameter}&\textbf{Value} \\
  \hline
  \hline
  Cell radius $R$ &0.5 km\\
  \hline
  D2D group radius $r$ & $60$ m \\
  \hline
  Pathloss model for cellular link & $128.1+37.6 \log (d [ \text{km}])$\\
  \hline
  Pathloss model for D2D links& $148+40 \log (d [\text{km}])$ \\
  \hline
  Noise Power $\sigma_N^2$& $-114$ dBm\\
  \hline
  Uplink cellular user & \multirow{2}{*}{ $24$ dBm} \\
  transmit power $p^t_j, j \in \mathcal{C}_u$& \\
  \hline
  D2D link transmit power $p^t_j, j \in \mathcal{D}$& $24$ dBm\\
  \hline
  BS transmit power $P_B$ & $46$ dBm\\
  \hline
  Minimum SINR requirement $\xi^{\text{min}}_{j}$ & $0$ dB \\
  \hline
  \multirow{2}{*}{Multipath fading}&Rayleigh fading \\
  & with unit variance\\
  \hline
  \multirow{2}{*}{Shadowing}&Log-normal distribution \\
  & with standard deviation of $8$ dB \\
  \hline
  Minimum success transmit & \multirow{2}{*}{ $99\%$}\\
  probability requirement $\psi^{\text{min}}_{j}$ & \\
  \hline
  \end{tabular}
\end{table}

In this section, we will provide simulation results to evaluate the proposed algorithms, as well as provide design insights for D2D communications. We will first compare different algorithms in the full CSI case, and then different partial CSI scenarios will be compared. A single cell scenario is considered, where both the uplink and downlink cellular users are uniformly distributed in the cell. We adopt the group distribution model in \cite{D2Dmodel} for D2D links, where each D2D transmitter and its associated D2D receiver are uniformly distributed in a randomly located group within the cell. We assume that different D2D pairs are distributed in different groups and the locations of different groups are independent. Equal power allocation is adopted in the downlink for cellular users. The default simulation parameters are summarized in Table \ref{tab1}. In addition, we set weights of all links as $1$.

\subsection{Weighted Sum-Rate Maximization with Full CSI}

\begin{figure}[t]
\centering
\includegraphics[width=3in]{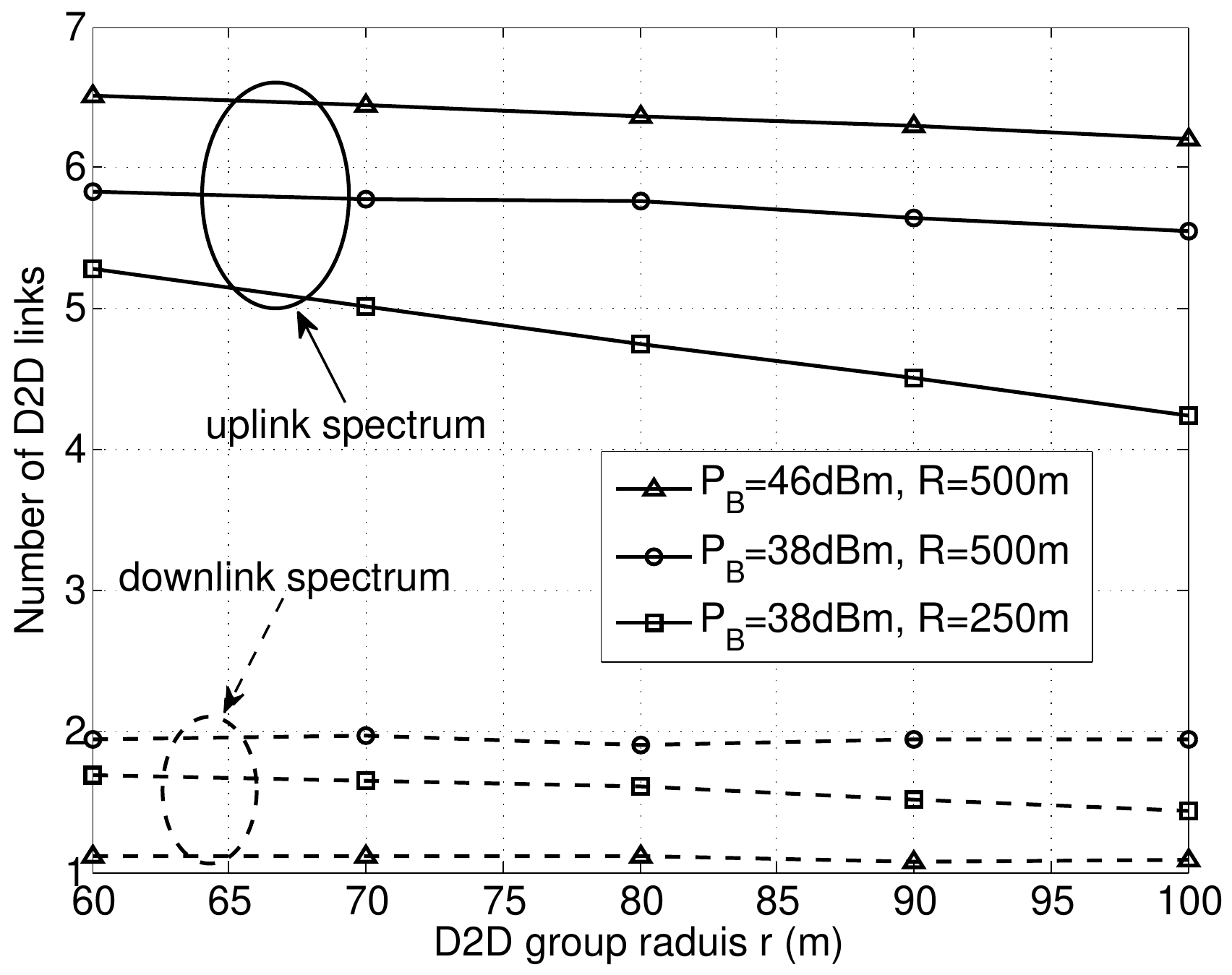}
\caption{Comparison of the number of D2D links sharing the uplink spectrum and downlink spectrum with $M_u=N_{uc}=4$, $M_d=N_{dc}=4$ and $N_d=8$. The optimal DP algorithm is applied.}
\label{uplink}
\end{figure}

\begin{figure}[t]
\centering
\includegraphics[width=3in]{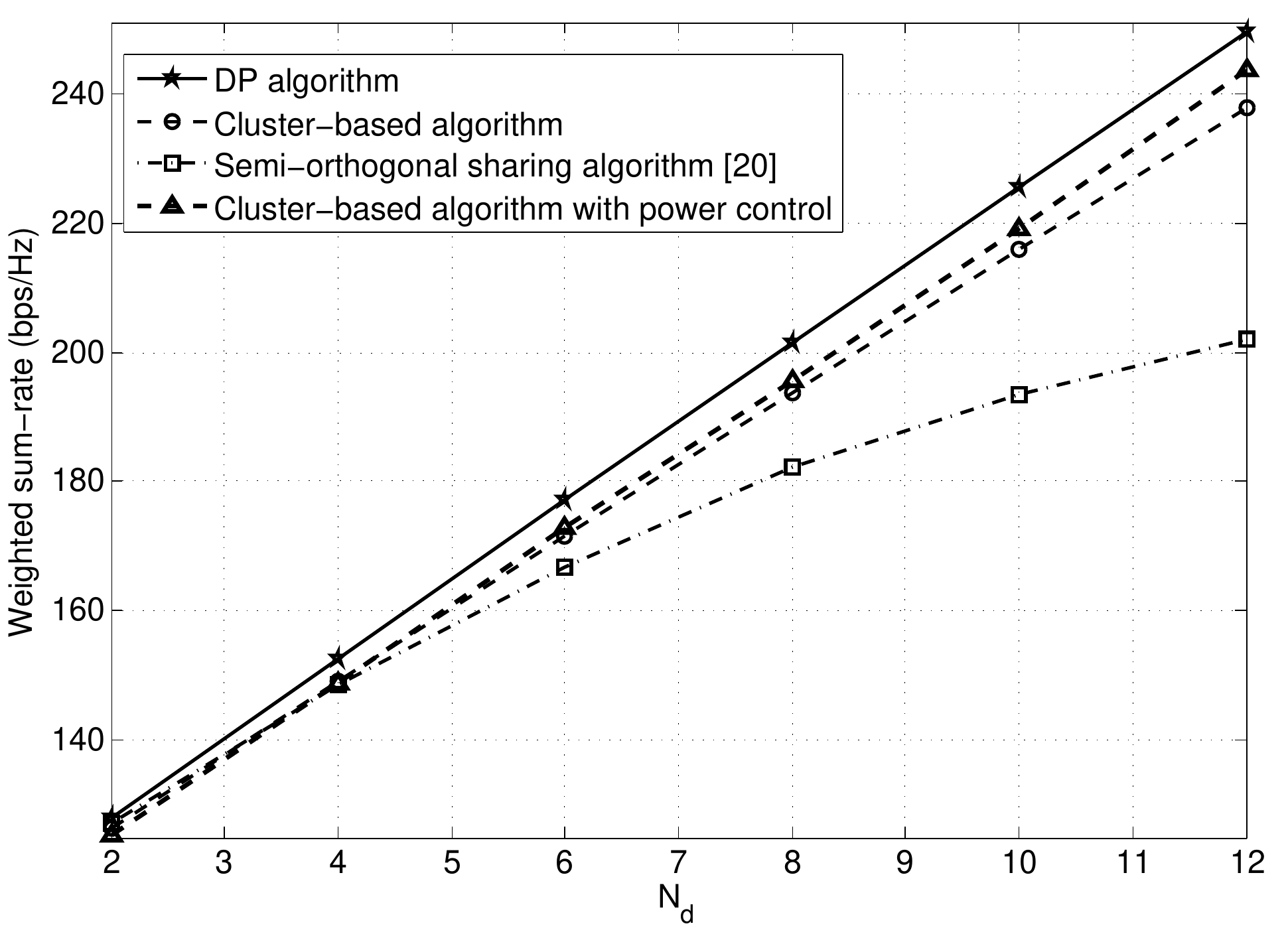}
\caption{Comparison of the weighted sum-rate for different algorithms with $M_u=N_{uc}=4$ and $M_d=N_{dc}=4$.}
\label{fig_scheme}
\end{figure}

\begin{figure}[t]
\centering
\includegraphics[width=3in]{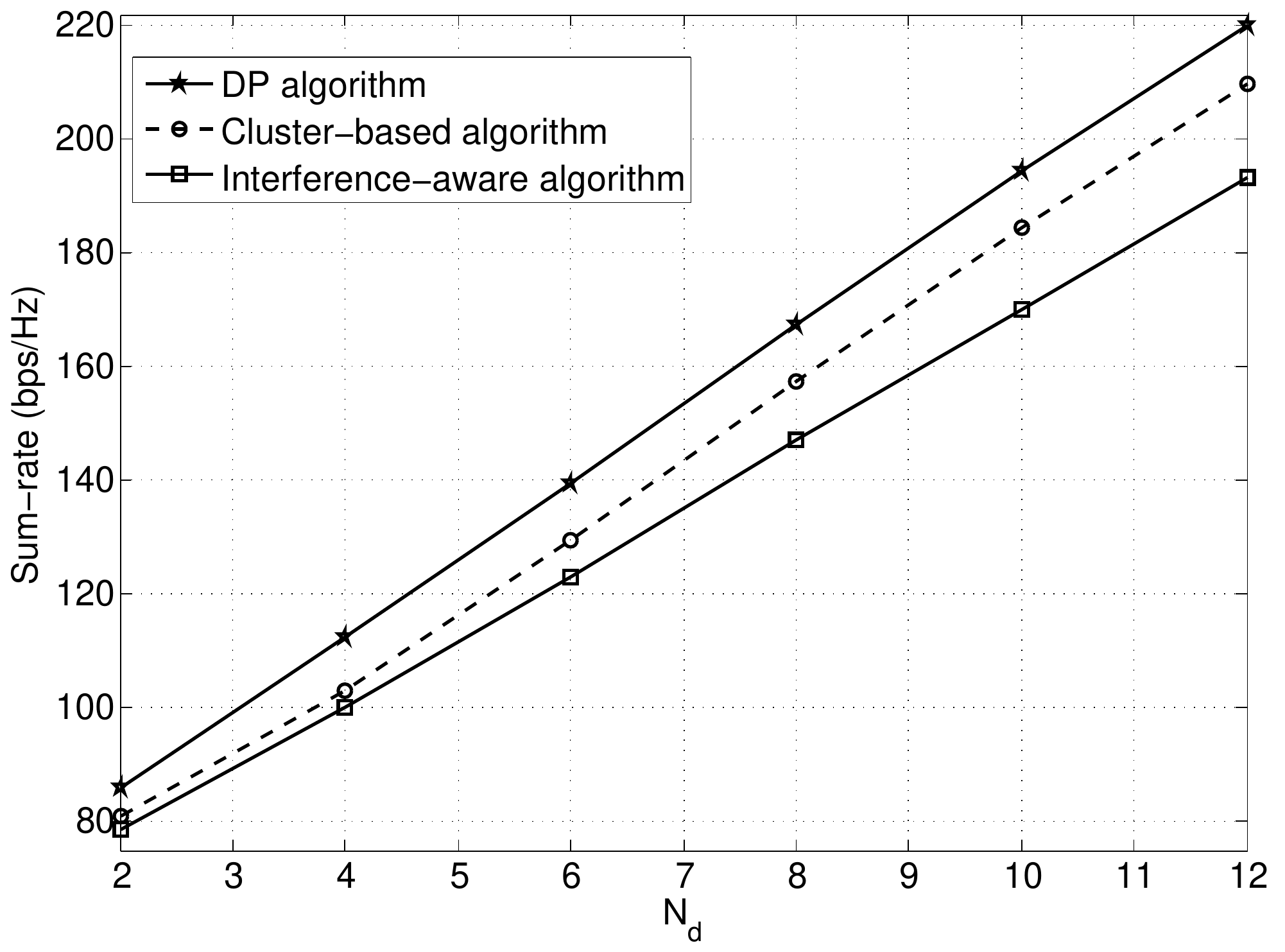}
\caption{Comparison of the weighted sum-rate for different algorithms with $M_u=N_{uc}=0$ and $M_d=N_{dc}=4$.}
\label{fig_compare}
\end{figure}

\begin{figure}[t]
\centering
\includegraphics[width=3in]{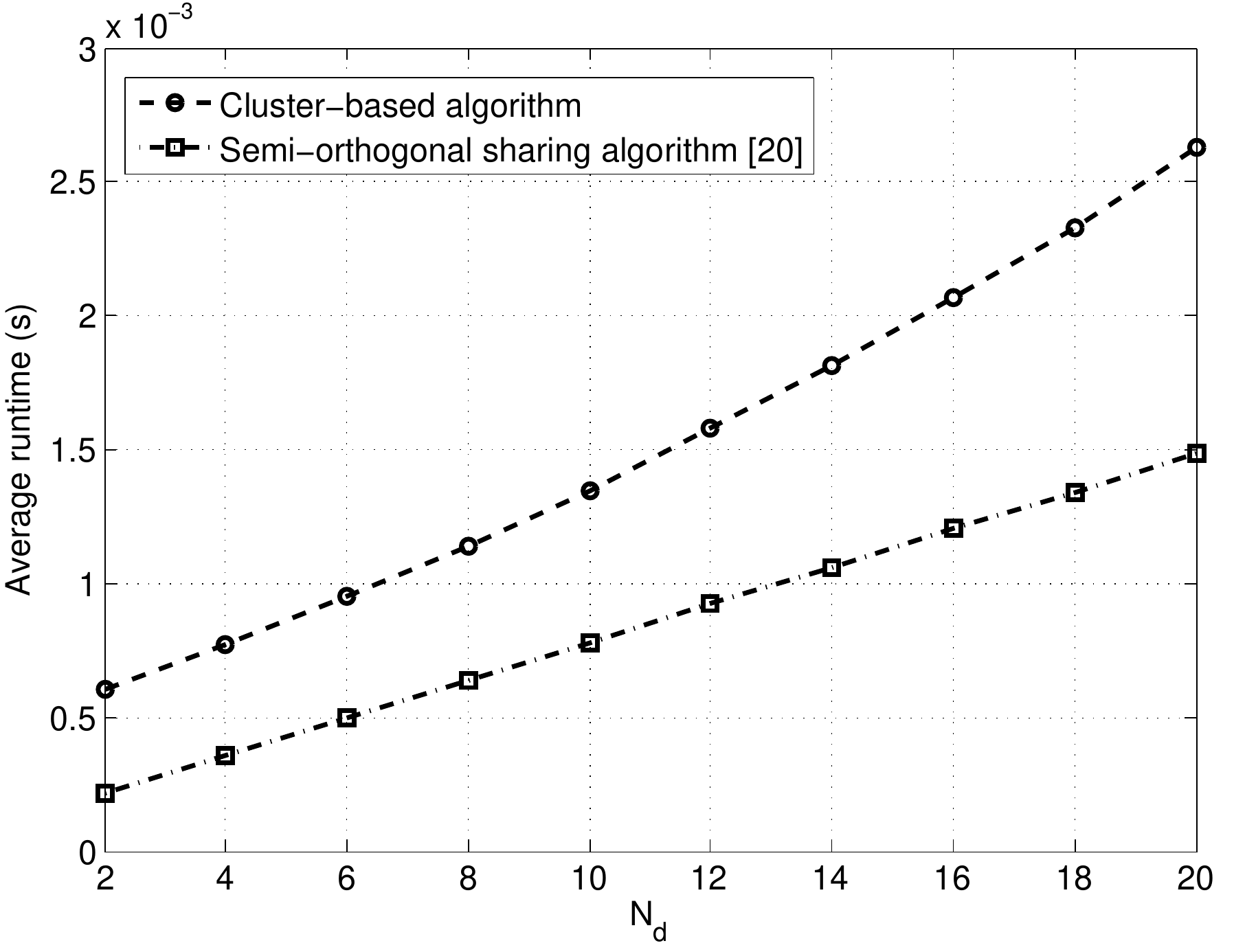}
\caption{Comparison of the average runtime for different algorithms with $M_u=N_{uc}=4$ and $M_d=N_{dc}=4$.}
\label{time}
\end{figure}

We will first investigate whether the uplink or downlink spectrum is accessed more frequently by the D2D links. Fig. \ref{uplink} compares the number of D2D links sharing the uplink and downlink spectrums, and it is shown that more D2D links will share the uplink spectrum. Interestingly, when we lower the BS transmit power, the number of  D2D links sharing the uplink spectrum decreases, while that sharing the downlink spectrum increases. Intuitively, the BS is always much more powerful than the devices. In the downlink spectrum, the BS may generate harmful interference to the D2D links sharing the same spectrum. Thus, it is better for the D2D links to share the uplink spectrum in most cases.

The advantage of allowing multiple D2D links to share the same channel is demonstrated in Fig. \ref{fig_scheme} by comparing with the previous study in \cite{one2one}, which allowed at most one D2D link to share the same channel and will be referred as the \emph{semi-orthogonal sharing algorithm}. In Fig. \ref{fig_scheme}, we find that the weighted sum-rate gain between our proposed sub-optimal algorithm and the semi-orthogonal sharing algorithm increases with the number of D2D links. This means that when the number of D2D links in the cell becomes larger, the proposed algorithm can significantly improve the performance by allowing multiple D2D links to share the same channel. We also evaluate the effect of power control in Fig. \ref{fig_scheme} by applying the geometric programming based power control algorithm in \cite{powercontrol} after the greedy channel assignment. It was shown in \cite{powercontrol} that such a power control algorithm can achieve or come very close to the global optima with very high probability. With power control, we see that the performance is slightly improved. A joint consideration of channel assignment and power allocation will provide additional performance gain, which is beyond the scope of this paper and will be left to our future work.

Fig. \ref{fig_compare} compares our proposed algorithm with the interference-aware algorithm \cite{nphard}, where multiple D2D links are allowed to share the same channel. Since the QoS requirements are not considered in \cite{nphard}, we set the minimum SINR requirements as 0. It is shown that our proposed cluster-based algorithm not only has the ability to deal with more general cases, but also outperforms the interference-aware algorithm.

The average runtime of the sub-optimal algorithm is shown in Fig. \ref{time}, which is also compared with a benchmark in \cite{one2one}. We see that the runtime of the proposed algorithm increases almost linearly with the number of D2D links, which is much better than the worst case complexity given in Section V-C. Even with 20 D2D links, the average runtime is only about twice that of the algorithm in \cite{one2one}, but with much better performance, as shown in Fig. \ref{fig_scheme}.

\subsection{Access Rate Maximization with Full CSI}

In this part, we will consider another objective, i.e., to maximize the access rate as defined in (\ref{weighted sr}). Fig. \ref{ar1} verifies the effectiveness of the proposed algorithms while maximizing the access rate by comparing the optimal DP algorithm and the cluster-based sub-optimal algorithm. We can find that the performance of the cluster-based algorithm is close to that of the DP algorithm, which indicates that our proposed cluster-based algorithm works well when maximizing the access rate. In order to show that different performance metrics will lead to different network designs, we also consider the access rate while maximizing the weighted sum-rate. We can find that the channel assignment that maximizes the weighted sum-rate cannot be directly applied to a network that needs to maximize the access rate.

\subsection{Investigation of Different Partial CSI Scenarios}
\begin{figure}[t]
\centering
\includegraphics[width=3in]{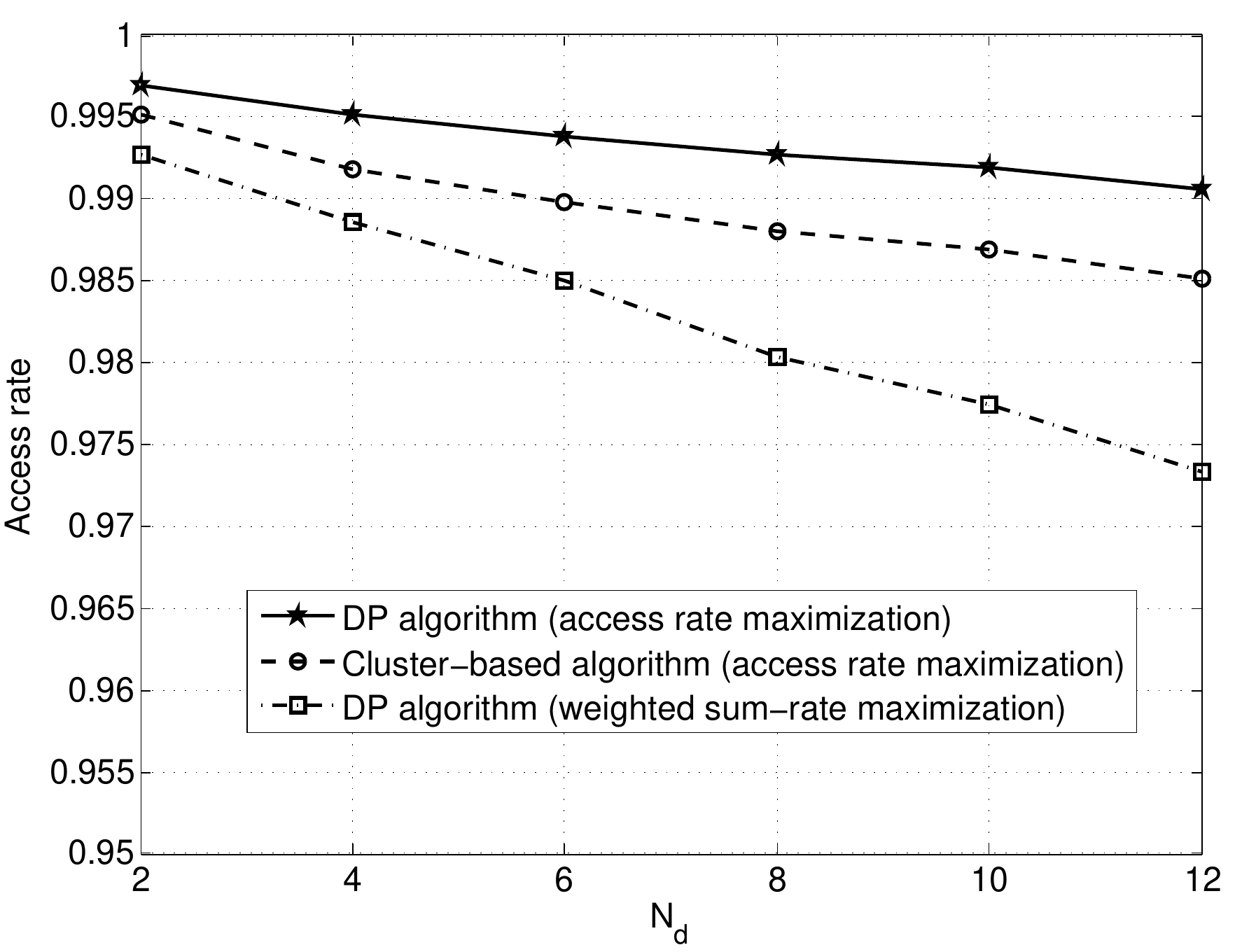}
\caption{Comparison of the access rate for different algorithms with $M_u=N_{uc}=4$ and $M_d=N_{dc}=4$.}
\label{ar1}
\end{figure}

\begin{figure}[t]
\centering
\includegraphics[width=3in]{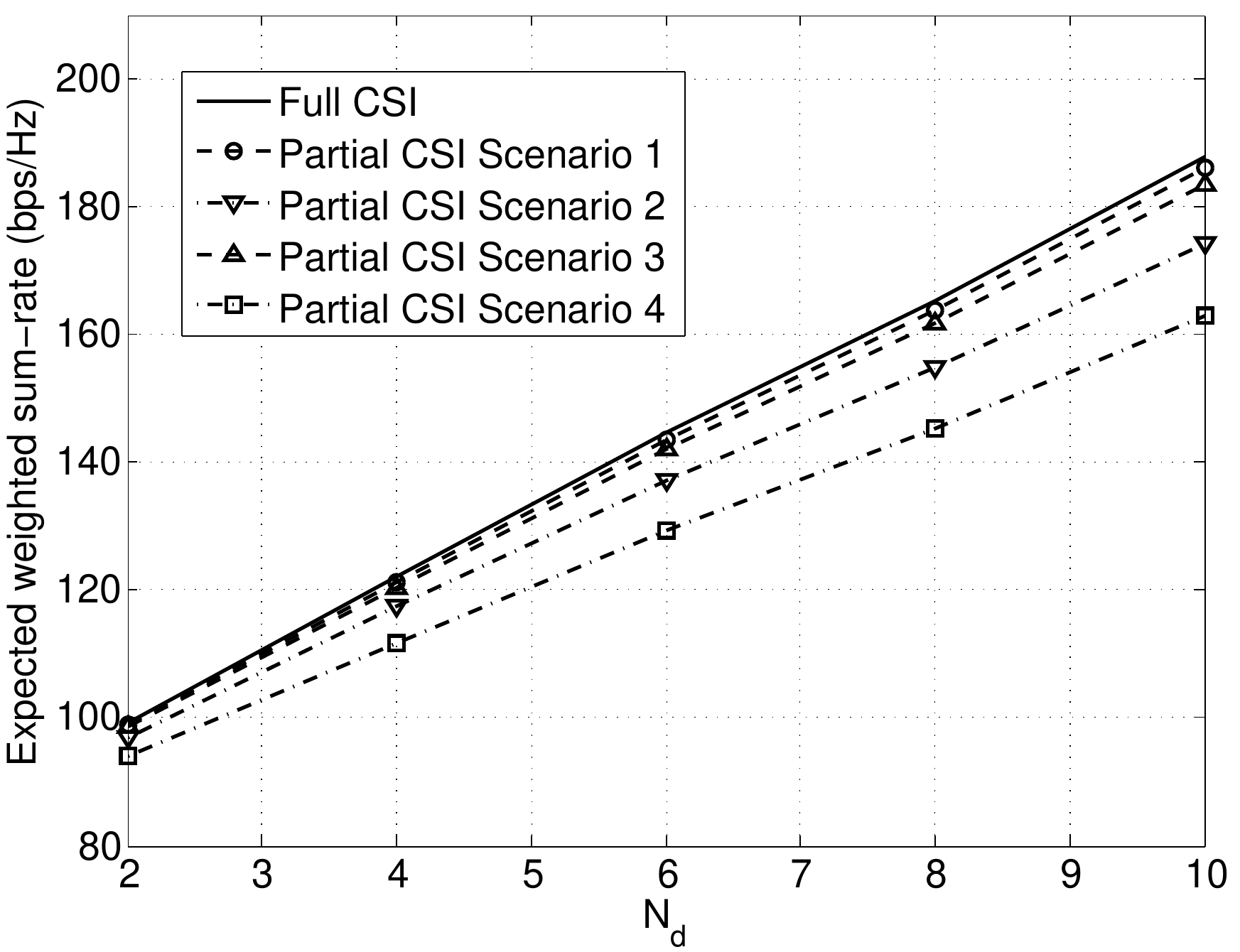}
\caption{Comparison among four scenarios with partial CSI when $M_u=N_{uc}=3$, $M_d=N_{dc}=3$ and cell radius $R=500m$.}
\label{partial1}
\end{figure}

\begin{figure}[t]
\centering
\includegraphics[width=3in]{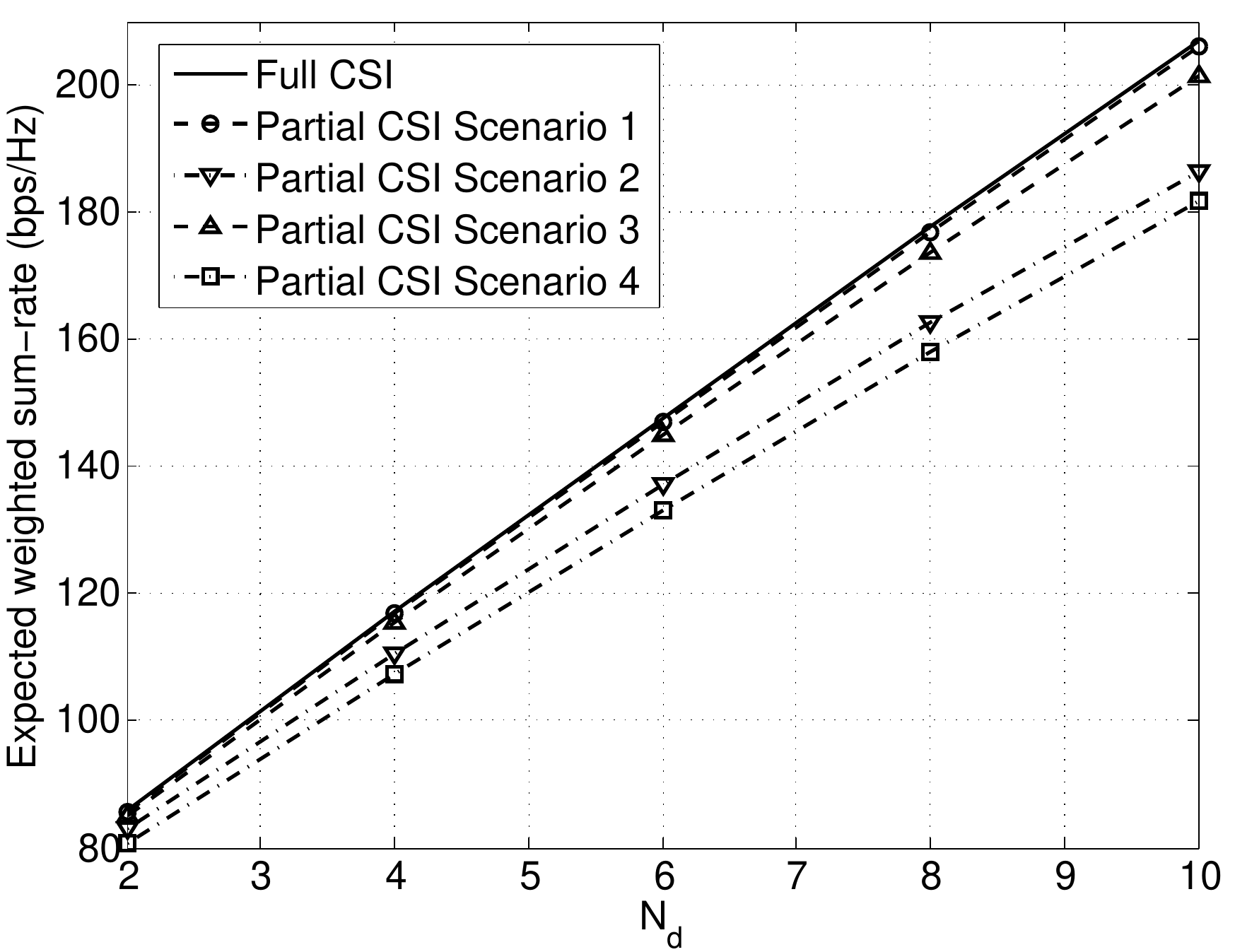}
\caption{Comparison among four scenarios with partial CSI when $M_u=N_{uc}=3$, $M_d=N_{dc}=3$ and cell radius $R=1km$.}
\label{partial2}
\end{figure}
\begin{figure}[t]
\centering
\includegraphics[width=3in]{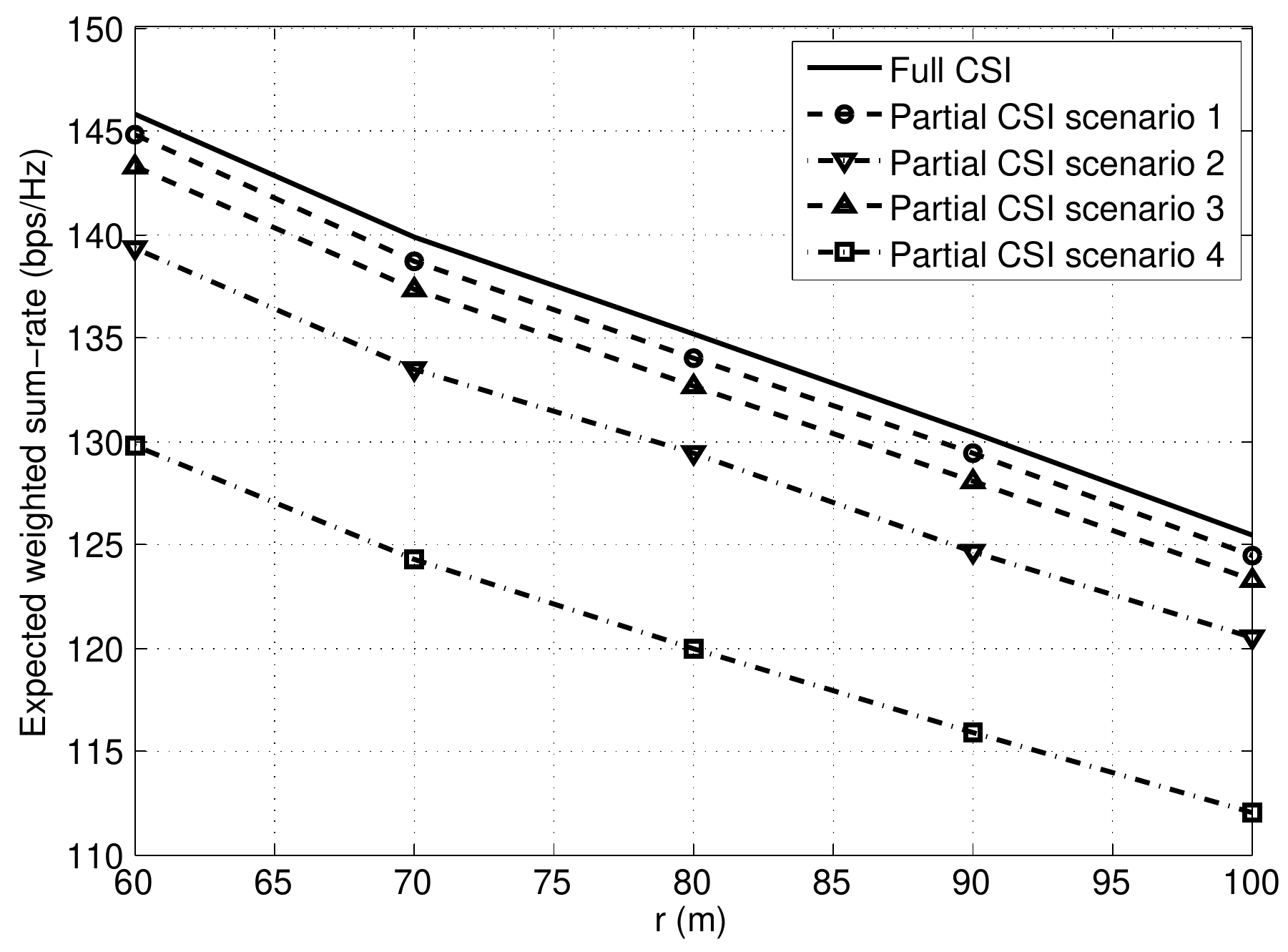}
\caption{Comparison among four scenarios with partial CSI when $M_u=N_{uc}=3$, $M_d=N_{dc}=3$, $N_d=6$, and cell radius $R=500m$. Meanwhile, Ricean fading channel, with a $K$-factor as $3$ dB, is assumed for each D2D link, and Rayleigh fading channels are assumed for cellular and interference links.}
\label{partial3}
\end{figure}

In this part, we will demonstrate the relative importance of the CSI of different links by comparing different partial CSI scenarios. The simulation results are shown in Fig. \ref{partial1} to Fig. \ref{partial4}, where the four partial CSI scenarios introduced in Section III-A are compared using the optimal DP algorithm, with different numbers of D2D links, different cell radii and different D2D group radii. The maximum weighted sum-rate in the full CSI scenario is also shown as the performance benchmark. Firstly, recall that the only difference between Scenarios 1 and 2 is whether the BS has knowledge of the small-scale fading gains of the D2D communication links. Since the gap between Scenarios 1 and 2 is relatively large, we can conclude that \emph{the CSI of the D2D communication links has a significant effect on the performance}. Secondly, recall that the only difference between Scenarios 1 and 3 is whether the BS knows the small-scale fading gains of the interference links from the BS to the D2D receivers. We see that the performance of Scenario 3 is close to that of Scenario 1. Thus, \emph{the D2D receivers may not report the CSI of the interference links from the BS to the D2D receivers}. As shown in Section VI-A, the D2D links prefer to share the uplink spectrum, and, thus, \emph{the CSI of the interference links in the downlink spectrum is not critical}. Finally, the only difference between Scenarios 3 and  4 is whether the BS has knowledge of the CSI of the interference links from the D2D transmitters to the BS. Since the gap between Scenarios 3 and  4 is relatively large, \emph{the CSI of the interference links from the D2D transmitters to the BS has a significant effect on the performance}. As a conclusion, the expected weighted sum-rate in both Scenarios 1 and 3 is close to that of the full CSI scenario, and Scenario 3 is a good choice to balance the network overhead and performance, i.e., we should try to obtain the CSI in practical D2D networks as in Scenario 3.

According to the results in Figs. \ref{partial1} - \ref{partial4}, such a conclusion still holds when we change different parameters, namely, the number of D2D links, the cell radius, the D2D group radius, and channel fading model. In addition, when the cell radius becomes larger, as in Fig. \ref{partial2} and Fig. \ref{partial4}, the gap between the full CSI scenario and Scenario 2, where the CSI of the D2D communication links cannot be acquired by the BS, becomes larger. Thus, when the cell radius increases, the CSI of the D2D communication links becomes more important for the network performance. This is because, for a larger cell radius, more D2D links can be served.

As Scenario 3 has been shown to be a favorable partial CSI scenario, we further compare the performance of the cluster-based sub-optimal algorithm with the optimal DP algorithm in Fig. \ref{partial5}. It is shown that the performance of the cluster-based sub-optimal algorithm is very close to that of the optimal algorithm, which indicates that it can also provide near-optimal performance in the partial CSI case.

\begin{figure}[!t]
  \centering
  \includegraphics[width=3in]{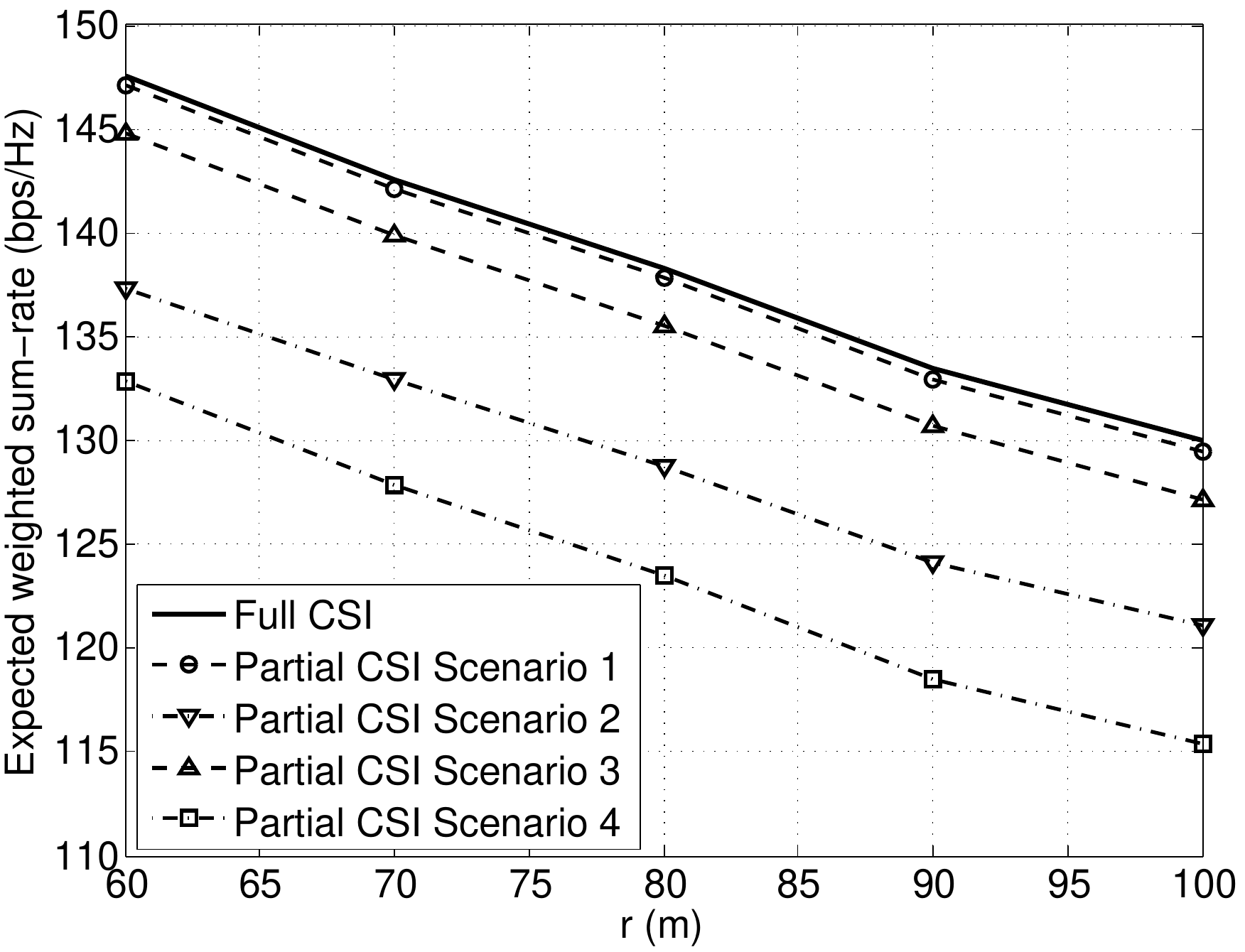}
\caption{Comparison among four scenarios with partial CSI when $M_u=N_{uc}=3$, $M_d=N_{dc}=3$, $N_d=6$ and cell radius $R=1km$.}
  \label{partial4}
\end{figure}
\begin{figure}
  \centering
  \includegraphics[width=3in]{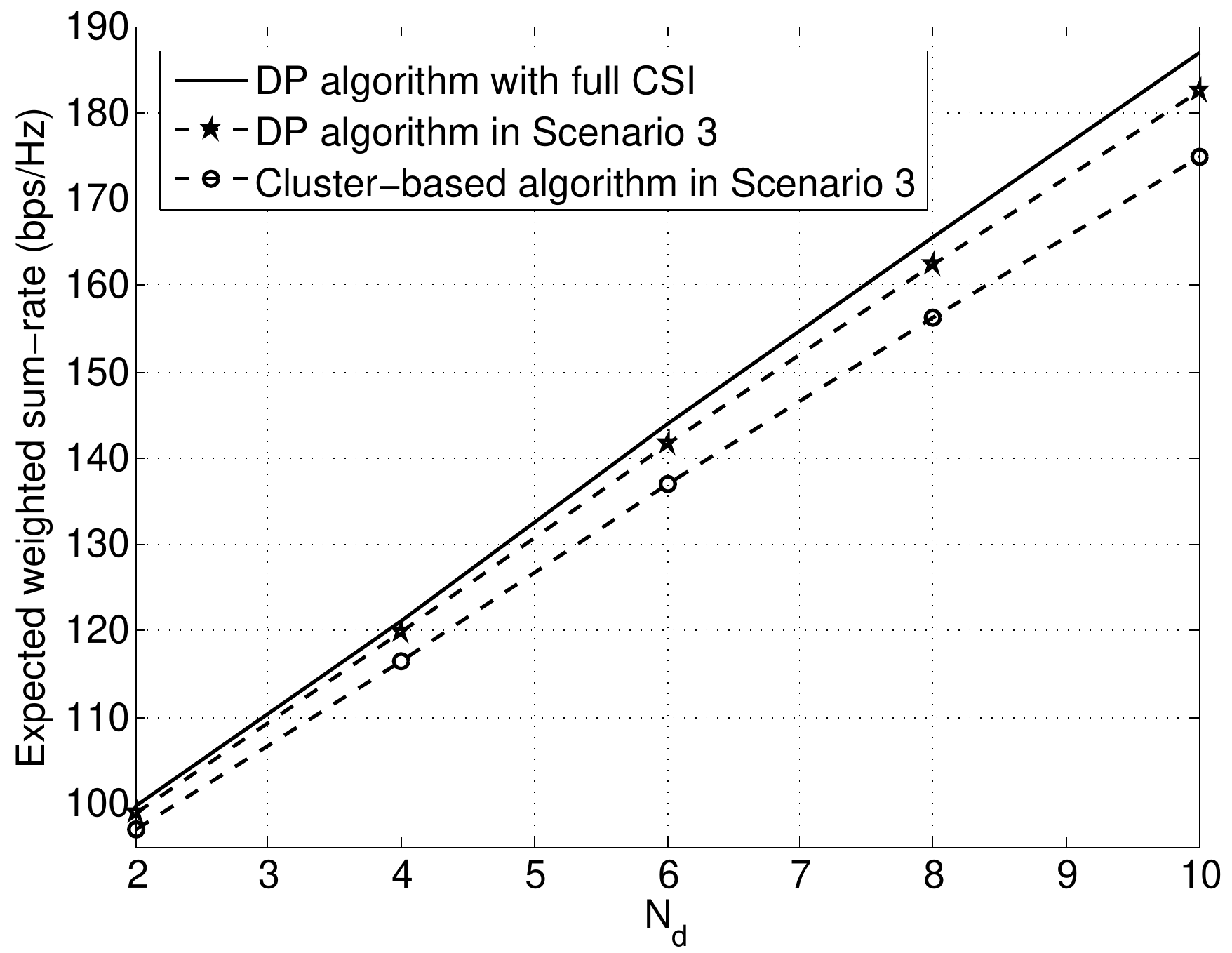}
  \caption{Comparison of the weighted sum-rate for different algorithms with $M_u=N_{uc}=3$ and $M_d=N_{dc}=3$.}
  \label{partial5}
\end{figure}

\section{conclusions}
In this paper, we investigated the channel assignment problem with partial CSI for D2D communications where multiple D2D links are allowed to share the same channel. Meanwhile, the minimum successful transmission probability requirements for both cellular and D2D links are enforced. An optimal DP algorithm and a cluster-based sub-optimal algorithm were proposed. To avoid high-complexity numerical integrations, we derived closed-form expressions for the expected weighted sum-rate and successful transmission probabilities. Simulation results demonstrated that our proposed cluster-based algorithm provides performance close to that of the optimal algorithm. Furthermore, we found that D2D links are more likely to use the uplink spectrum compared to the downlink spectrum. By comparing four different scenarios with partial CSI, we observed that the knowledge of the CSI of the D2D communication links and the interference links from the D2D transmitters to the BS has a significant effect on the network performance, while the knowledge of the CSI of the interference links from the BS to the D2D receivers does not significantly affect the network performance. For future works, it would be interesting to extend the proposed algorithms to the multi-cell scenario, and investigate distributed algorithms for practical implementation. It will also be interesting to allow each user to access multiple channels, for which the proposed DP algorithm can be extended with increased complexity, while new heuristic sub-optimal algorithms will be needed. Joint mode selection and channel assignment is also important, with some initial results reported in \cite{jointmode}.

\section*{Appendix}
\subsection{Proof of Proposition 1}
From definitions, when $|\mathcal{L}'_i|>0$, the successful transmission probability and expected rate are 
\begin{align} \label{proborigin}
&p^s_{i,j} = \notag \\
&\begin{cases}
\int^{\eta_{i,j}}_0 f_{Y_{i,j}\left( \mathcal{L}'_i \right)} \left( y \right) dy & \text{if } \beta_{i,j,j} \text{ is known}, \\
\int_{\frac{\xi^{\text{min}}_{j} \nu }{ \lambda_{j,j}}}^{\infty}\int_{0}^{\iota_j} f_{Y_{i,j}\left( \mathcal{L}'_i \right)} \left( y \right) dy f_{j} \left( x \right) dx & \text{if $\beta_{i,j,j}$ is unknown},
\end{cases}
\end{align}
and
\begin{align} \label{rateorigin}
r_{i,j}= \left\{
\begin{array}{ll}
\int^{ \eta_{i,j}}_0 \log \left(  1+ \frac{\lambda_{j,j}\beta_{i,j,j} }
{\nu  + y} \right) & \multirow{2}{*}{$\text{if } \beta_{i,j,j} \text{ is known},$} \\
\times f_{Y_{i,j}\left( \mathcal{L}'_i \right)} \left( y \right) dy  &  \\
\int_{\frac{\xi^{\text{min}}_{j} \nu }{ \lambda_{j,j}}}^{\infty}\int_{0}^{\iota_j} \log \left(  1+ \frac{\lambda_{j,j}\beta_{i,j,j} }
{\nu  + y} \right) & \multirow{2}{*}{$\text{if $\beta_{i,j,j}$ is unknown}.$} \\
\times f_{Y_{i,j}\left( \mathcal{L}'_i \right)} \left( y \right) dy f_{j} \left( x \right) dx &
\end{array}
\right.
\end{align}
As interference links are Nakagami fading, $Y_{i,j}\left( \mathcal{L}'_i \right)$ is a sum of Gamma random variables. From \cite{sumofgamma}, we can get its pdf as
\begin{align} \label{epdfgamma}
& f_{Y_{i,j}\left( \mathcal{L}'_i \right)} \left( y \right) =
\prod \limits_{z \in \mathcal{L}'_i} \left( \frac{\lambda_{z,j}}{m_{z,j}} \right)^{m_{z,j}} \notag \\
&\times \sum \limits_{n=0}^{+ \infty} \frac{\delta_n y^{\rho+n-1} (\theta^{\max}_j)^{n}  \exp(-y\theta^{\max}_j)}{\Gamma \left(\rho+n \right)}, y \geqslant 0,
\end{align}
where
\begin{align} \label{refdelta}
\delta_0= &1, \notag \\
\delta_{n}= &\frac{1}{n} \sum \limits_{l=1}^{n} \left[ \sum \limits_{z \in \mathcal{L}'_i} m_{z,j} \left( 1- \frac{\lambda_{z,j}}{\theta^{\max}_j m_{z,j}} \right)^l \right] \delta_{n-l}, \notag \\
& n=1,2,\cdots.
\end{align}
In \cite{sumofgamma}, the infinite series in (\ref{epdfgamma}) is proved to be uniform convergence and can be accurately approximated by finite terms in practice. Substituting (\ref{epdfgamma}) into (\ref{proborigin}) and applying Eqn. (3.381-1) in \cite{int}, we can get the successful transmission probability as in (\ref{propnakagami}). Substituting (\ref{epdfgamma}) into (\ref{rateorigin}), we can get the expected rate as in (\ref{expectedrate}). If $|\mathcal{L}'_i| = 0$, the SINR of D2D link $j$ becomes $\xi_{i,j} \left( \mathcal{L}'_i \right) =\frac{{\lambda_{j,j} \beta_{i,j,j} }}{{\nu  }}$, where only $\beta_{i,j,j}$ could be a random variable. Thus, it is straightforward to calculate $p^s_{i,j}$ and $r_{i,j}$.

\subsection{Proof of Lemma 1}
When $\beta_{i,j,j}$ is known, according to (\ref{refmu}), $\mu_0=\int_0^{\eta_{i,j}} \ln\left( 1+ \frac{\lambda_{j,j}\beta_{i,j,j}}{\nu+y}\right) e^{ -y \theta_j^{\max} } dy$. Using the integration by parts and applying Eqn. (3.352-1) in \cite{int}, we can get its closed-form expression as in (\ref{closedmu}), where 
$D(t_1,t_2)=\left( \theta_j^{\max} \right)^{-1}\big\{\ln(t_2) - \ln (t_1+t_2) e^{-t_1 \theta_j^{\max}}+ e^{t_2 \theta^{\max}_j } \big[ \text{Ei} \left(- (t_1+t_2)\theta^{\max}_j \right) - \text{Ei} (-t_2 \theta^{\max}_j ) \big]\big\}.$
Using the integration by parts, the recurrence relation of $\mu_k$ can be obtained as in (\ref{closedmu}), where
\begin{equation}
A_k(t_1,t_2)=\ln \left(1+ \frac{t_2}{t_1 +\nu} \right) t_1^k \exp \left( -t_1  \theta^{\max}_j \right),
\end{equation}
and $E_k(t_1,t_2)= \int_0^{t_1} \frac{y^n}{t_2+y} \exp (-y \theta_j^{\max}) dy$. By binomial theorem, we can rewrite $E_k(t_1,t_2)$ as
\begin{align}
E_k & (t_1,t_2)= \int_0^{\infty} \frac{y^k e^{-y \theta_j^{\max}}}{t_2+y} dy-e^{t_2\theta_j^{\max}} \sum \limits_{l=0}^{k} \binom{k}{l} t_1^{k-l} \notag \\
& \times \int_{t_1+t_2}^{\infty} \frac{(y'-t_1-t_2)^l}{y'} \exp (-y' \theta_j^{\max}) dy'
\end{align}
By Eqn. (3.383-9) and Eqn. (3.383-10) in \cite{int}, the closed-form of $E_k(t_1,t_2)$ can be obtained as
\begin{align*}
E_k & (t_1,t_2) = e^{ \theta_j^{\max} t_2} \bigg[t_2^k \Gamma(k+1) \Gamma(-k,t_2 \theta^{\max}_j) \notag \\
&-\sum \limits_{l=0}^{k} \binom{k}{l} t_1^{k-l} \left( t_1+t_2 \right)^l \Gamma(l+1) \Gamma \left(-l,(t_1+t_2)\theta^{\max}_j \right) \bigg].
\end{align*}
When $\beta_{i,j,j}$ is unknown, $\mu_k$ can be derived similarly. 

\subsection{Proof of Corollary 1}
Under Rayleigh fading, $Y_{i,j} \left( \mathcal{L}'_i \right)$ can be written as a linear combination of exponential random variables. Then, based on \cite{expr}, the pdf of $Y_{i,j} \left( \mathcal{L}'_i \right)$ is
\begin{align} \label{epdfce}
& f_{Y_{i,j}\left( \mathcal{L}'_i \right)} \left( y \right) = \notag \\
& \sum \limits_{z \in \mathcal{L}'_i} \left( \prod \limits_{k \in \mathcal{L}'_i, k \ne z} \left( \lambda_{z,j}-\lambda_{k,j}\right) ^{-1}\right) \lambda_{z,j}^{|\mathcal{L}'_i|-2} e^{-y/\lambda_{z,j}} \text{ if } y \geqslant 0.
\end{align}
By substituting (\ref{epdfce}) into (\ref{proborigin}) and (\ref{rateorigin}), $p^s_{i,j}$ can be derived as (\ref{p1_pr}), and the expected rate is
\begin{align}
r_{i,j}=&\sum \limits_{z \in \mathcal{L}'_i} \left( \prod \limits_{k \in \mathcal{L}'_i, k \neq z} \left( \lambda_{z,j}-\lambda_{k,j}\right) ^{-1}\right) \lambda_{z,j}^{|\mathcal{L}'_i|-2} \notag \\
& \times \int^{\eta_{i,j}}_0 \log \left(  1+ \frac{\lambda_{j,j}\beta_{i,j,j} }
{\nu  + y} \right) e^{-y/\lambda_{z,j}} dy. \label{c-d-e1}
\end{align}
Then, by partial integration and Eqn. (3.352-1) in \cite{int}, we can get the closed-form expression of the expected rate.

\subsection{Proof of Corollary 2}
Substituting (\ref{epdfce}) into (\ref{proborigin}), the success probability can be derived as (\ref{p2_pr}). Meanwhile, the expected rate can be written as
\begin{align*}
r_{i,j}
=& \sum \limits_{z \in \mathcal{L}'_i} \left( \prod \limits_{k \in \mathcal{L}'_i, k \ne z} \left( \lambda_{z,j}-\lambda_{k,j}\right) ^{-1}\right) \lambda_{z,j}^{|\mathcal{L}'_i|-2} \notag \\
& \times \int_0^{\infty}\int_{\frac{\xi^{\text{min}}_{j} \left( \nu + y \right)}{\lambda_{j,j}} }^{\infty}\log \left( 1+\frac{\lambda_{j,j}x}{\nu +y} \right) e^{ \left( -\frac{y}{\lambda_{z,j}} -x \right)} dx dy.
\end{align*}
By replacing $y$ with $t=\frac{\lambda_{j,j}x}{\xi^{\text{min}}_{j}\left(\nu +y\right)}$, and applying partial integration and Eqn. (3.352-2) in \cite{int}, we can get the closed-form expression of the expected rate for $|\mathcal{L}'_i| >0$.
If $|\mathcal{L}'_i| = 0$, the SINR of D2D link $j$ becomes
$\xi_{i,j} \left( \mathcal{L}'_i \right) =\frac{{\lambda_{j,j} \beta_{i,j,j} }}
{{\nu  }}$,
where only $\beta_{i,j,j}$ is a random variable. Then, by Eqn. (3.352-2) in \cite{int}, we can easily get the successful transmission probability and expected rate for D2D link $j$ using channel $i$ for $|\mathcal{L}'_i| =0$.

\bibliographystyle{IEEEtran}
\bibliography{IEEEabrv,ref}

\end{document}